\documentclass[a4paper,UKenglish,cleveref, autoref, thm-restate]{lipics-v2021}

\pdfoutput=1 
\hideLIPIcs

\usepackage{graphicx} 
\usepackage{amsmath, amssymb}
\usepackage{amsthm}
\usepackage{url}
\usepackage{hyperref}
\usepackage{caption}
\usepackage{makeidx}
\usepackage{mathtools}
\usepackage{todonotes}
\usepackage{makecell}
\usepackage{rotating}
\usepackage{colortbl}
\usepackage{xcolor}

\colorlet{MediumRed}{red!80!white}
\colorlet{MediumBlue}{blue!65!white}
\colorlet{MediumGreen}{blue!55!green}

\usepackage[ruled, vlined, linesnumbered]{algorithm2e}

\usepackage{tikz}
\usetikzlibrary{quotes}
\tikzstyle{vertex}=[circle, draw, inner sep=0pt, minimum size=6pt]
\newcommand{\vertex}{\node[vertex]}

\usetikzlibrary{arrows.meta, decorations.markings}
\usetikzlibrary{shapes.geometric}
\usetikzlibrary{trees, positioning}

\tikzset{
    mynode/.style={
        draw,
        ellipse,
        font=\sffamily\small,
        minimum width=2.5em,
        minimum height=1.4em,
        inner sep=0pt,
        align=center,
        thick
    }
}

\usepackage{xspace}
\newcommand{\A}{\ensuremath{\texttt{A}}}
\newcommand{\C}{\ensuremath{\texttt{C}}}
\newcommand{\G}{\ensuremath{\texttt{G}}}
\newcommand{\T}{\ensuremath{\texttt{T}}}
\newcommand{\xx}{\cellcolor{lightgray}}

\newcommand{\greedyNeckAlgo}{\textsc{necklaceCover}\xspace}
\newcommand{\implementedNeckCover}{\texttt{matchNecklaceCover}\xspace}
\newcommand{\greedyBaseline}{\texttt{greedyBaseline}\xspace}

\newcommand{\gp}[1]{{#1}}
\newcommand{\rg}[1]{}
\newcommand{\cpm}[1]{}

\newcommand{\dBGedge}[1]{\ensuremath{\rightarrow_{\texttt{#1}}}\xspace}

\title{Variations on the Problem of Identifying Spectrum-Preserving String Sets} 

\author{Sankardeep Chakraborty}{University of Tokyo}{sankardeep.chakraborty@gmail.com}{https://orcid.org/0000-0002-2395-4160}{}
\author{Roberto Grossi}{University of Pisa}{roberto.grossi@unipi.it}{https://orcid.org/0000-0002-7985-4222}{}
\author{Ren Kimura}{University of Tokyo}{renkimura@g.ecc.u-tokyo.ac.jp}{}{}
\author{Giulia Punzi}{University of Pisa}{giulia.punzi@unipi.it}{https://orcid.org/0000-0001-8738-1595}{} 
\author{Kunihiko Sadakane}{University of Tokyo}{sada@mist.i.u-tokyo.ac.jp}{https://orcid.org/0000-0002-8212-3682}{}
\author{Wiktor Zuba}{University of Warsaw}{w.zuba@mimuw.edu.pl}{https://orcid.org/0000-0002-1988-3507}{}

\authorrunning{S. Chakraborty, R. Grossi, R. Kimura, G. Punzi, K. Sadakane, and W. Zuba} 

\Copyright{S. Chakraborty, R. Grossi, R. Kimura, G. Punzi, K. Sadakane, W. Zuba} 



\category{} 

\relatedversion{} 

\acknowledgements{Work done while RG, GP, and WZ visited the University of Tokyo under the EU PANGAIA project (European Union’s Horizon 2020 Research and Innovation Staff Exchange programme under the Marie Skłodowska-Curie grant agreement No. 872539), and RK visited the University of Pisa.
GP is supported by the Italian Ministry of Research, under the complementary actions to the NRRP “Fit4MedRob - Fit for Medical Robotics” Grant (\# PNC0000007)
}

\EventEditors{John Q. Open and Joan R. Access}
\EventNoEds{2}
\EventLongTitle{42nd Conference on Very Important Topics (CVIT 2016)}
\EventShortTitle{CVIT 2016}
\EventAcronym{CVIT}
\EventYear{2016}
\EventDate{December 24--27, 2016}
\EventLocation{Little Whinging, United Kingdom}
\EventLogo{}
\SeriesVolume{42}
\ArticleNo{23}

\begin{document}

\nolinenumbers 
\maketitle

\begin{abstract}
In computational genomics, many analyses rely on efficient storage and traversal of $k$-mers, motivating compact representations such as spectrum-preserving string sets (SPSS), which store strings whose $k$-mer spectrum matches that of the input. Existing approaches, including Unitigs, Eulertigs and Matchtigs, model this task as a path cover problem on the deBruijn graph. 
We extend this framework from paths to branching structures by introducing necklace covers, which combine cycles and tree-like attachments (pendants). We present a greedy algorithm that constructs a necklace cover while guaranteeing, under certain conditions, optimality in the cumulative size of the final representation. 

Experiments on real genomic datasets indicate that the minimum necklace cover achieves smaller representations than Eulertigs and comparable compression to the Masked Superstrings approach, while maintaining exactness of the $k$-mer spectrum.
\end{abstract}

\setcounter{page}{0}
\newpage

\section{Introduction}
\label{section:intro}

In modern computational genomics, many downstream analyses, such as read mapping and variant calling, reference-guided assembly, and metagenomic screening, begin by locating short exact matches, called $k$-mers, to a reference genome before performing heavier inference. To support these atomic queries, one wants a representation of $k$-mers that is compact, fast to traverse, and easy to interface with existing tooling. \textit{Spectrum-preserving string sets} (SPSS) meet these requirements by storing a set of strings whose $k$-mers match those of the input and using a smaller number of characters, when possible~\cite{rahman2021disk, rahman2021representation, sladky2023masked}. In particular, no-repetition SPSS (also known as simplitigs~\cite{bvrinda2021simplitigs}) keeps each $k$-mer exactly once, avoiding duplication overhead in memory and simplifying indexing pipelines. 

SPSS representations are used for disk compression, static $k$-mer membership indices, and as a basis for the Spectral Burrows-Wheeler Transform (SBWT), which supports fast membership queries and further space reductions \cite{alanko2022succinct,rahman2021disk,rahman2021representation}. Recent work unifies SPSS, simplitigs~\cite{bvrinda2021simplitigs}, matchtigs~\cite{schmidt2023matchtigs}, and Eulertigs~\cite{schmidt2023eulertigs} under the broader framework of masked superstrings, providing a theoretical foundation for optimizing $k$-mer set representations for diverse bioinformatics applications \cite{sladky2023masked}. It should be noted, however, that masked superstrings do not provide exact SPSS representations as the superstring containing all input $k$-mers may introduce false positives, i.e., $k$-mers that do not occur in the original strings.

\smallskip \noindent \textbf{Background.}
The SPSS problem can be formulated as following~\cite{rahman2021representation}. Let $\Sigma$ be an alphabet equipped with an optional reverse-complement mapping.~\footnote{For example, $\Sigma=\{\mathtt{A},\mathtt{C},\mathtt{G},\mathtt{T}\}$ for DNA, and $\mathtt{A}$-$\mathtt{T}$ and $\mathtt{C}$-$\mathtt{G}$ are complements of each other. The reverse complement of, say, $\mathtt{ATGCAAT}$ is $\mathtt{ATTGCAT}$.} Given a positive integer $k$, a $k$-mer is any substring of length $k$. The spectrum of a set of strings $X$, denoted $\mathrm{spec}_k(X)$, is the set of all $k$-mers (and their reverse complements, depending on the domain application) that appear as substrings in at least one string of $X$.   Given $k$ and an input set of strings $I$ (each of length at least $k$) over $\Sigma$, a set of strings $S$ is an SPSS  if $\mathrm{spec}_k(I) = \mathrm{spec}_k(S)$, i.e.\ $S$ contains exactly the same set of $k$-mers as $I$, and no \textit{extra} ones. In this paper, we focus on simplitigs, or no-repetition SPSS; namely, each $k$-mer appears exactly \textit{once} in $S$. The aim is to minimize the \emph{weight} of an SPSS, which is defined as its cumulative length $w(S) = \sum_{z \in S} |z|$. 
\gp{For instance, an SPSS of input set 
$I = \{\T\G\G\A\C\G\G\G\A\C\G\G\C\A\T, \C\A\G\T\T\C\C, \C\G\G\T\C\G\T\T, \G\G\C\A\G\C\T\}$
for $k=3$ is given by 
$S = \{\texttt{CAGTTCC}, \texttt{TGGCT}, \texttt{CAA}, \texttt{AGCAT}, \texttt{GGGTCGGACGT}\}$ 
with $w(S) = 31$. Note how, even if the number of strings increased, each $k$-mer of $I$ is now uniquely represented in $S$: for instance, the repeated $k$-mer $\C\G\G$, occurring twice in the first string of $I$ and once in its third string, now only occurs once in $S$ (in its last string). 
}

One of the main tools for constructing an SPSS for $\mathrm{spec}_k(I)$
is the \textit{node-centric de~Bruijn graph (dBG)}.  Each node in the order-$k$ dBG represents a distinct $k$-mer in $\mathrm{spec}_k(I)$, and each directed edge indicates that the last $k-1$ symbols of the source $k$-mer are equal to the first $k-1$ symbols of the target $k$-mer (self-loops are allowed) (see Figure~\ref{fig:dBG-node-edge}). Consequently, each directed path of $\ell$ nodes in the dBG spells a string of length $\ell + k - 1$: the full $k$-mer of the first node followed by one additional symbol for each of the remaining $\ell-1$ nodes in the path.  For example, \gp{in the left of Figure~\ref{fig:dBG-node-edge}, the dBG path (edge labels shown as subscripts) 
$\texttt{TGG} \rightarrow_{\texttt{T}} \texttt{GGT} \rightarrow_{\texttt{C}} \texttt{GTC} \rightarrow_{\texttt{G}} \texttt{TCG} \rightarrow_{\texttt{G}} \texttt{CGG}\rightarrow_{\texttt{C}} \texttt{GGC}$ with $\ell = 6$
spells the string $\texttt{TGGTCGGC}$.}
\gp{Note that this path corresponds to a \emph{trail} (nodes can be repeated, but edges cannot) with $\ell = 6$ edges in the edge-centric dBG ($\texttt{TG} \rightarrow_{\texttt{G}} \texttt{GG} \rightarrow_{\texttt{T}} \texttt{GT} \rightarrow_{\texttt{C}} \texttt{TC} \rightarrow_{\texttt{G}} \texttt{CG} \rightarrow_{\texttt{G}} \texttt{GG} \rightarrow_{\texttt{C}} \texttt{GC}$), and vice versa.}

\begin{figure}
    \centering
\begin{tikzpicture}[xscale=0.8]
\node[mynode] (TGG) at (-0,0) {{\T\G}\G};
\node[mynode] (GGG) at (-1,1.7) {{\G\G}\G};
\node[mynode] (GGC) at (0,-1.5) {\G\G\C};
\node[mynode] (GCA) at (0,-2.5) {\G\C\A};
\node[mynode] (CAA) at (0.3,-3.5) {\C\A\A};
\node[mynode] (CAT) at (-1,-3.5) {\C\A\T};
\node[mynode] (GGA) at (2,1.7) {\G\G\A};
\node[mynode] (CGG) at (2,0.5) {\C\G\G};
\node[mynode] (GGT) at (3.5,-1) {\G\G\T};
\node[mynode] (GAC) at (4,1.7) {\G\A\C};
\node[mynode] (ACG) at (4,0.5) {\A\C\G};
\node[mynode] (AGC) at (2,-2.5) {\A\G\C};
\node[mynode] (GCT) at (2,-1.5) {\G\C\T};
\node[mynode] (CAG) at (2,-3.5) {\C\A\G};
\node[mynode] (GTC) at (3.5,-2.5) {\G\T\C};
\node[mynode] (TCC) at (5.25,-2.5) {\T\C\C};
\node[mynode] (TCG) at (6,-0.5) {\T\C\G};
\node[mynode] (CGT) at (6,1) {\C\G\T};
\node[mynode] (GTT) at (7.5,-1.6) {\G\T\T};
\node[mynode] (TTC) at (6.75,-2.5) {\T\T\C};
\node[mynode] (AGT) at (7.5,-3.5) {\A\G\T};

\draw[-Stealth, label] (TGG) to node[left]{$\C$} (GGC);
\draw[-Stealth] (TGG) to node[above, sloped, xshift = 0.75em]{$\A$} (GGA);
\draw[-Stealth] (TGG) to node[below, sloped, xshift = 0.3em, yshift = 0.1em]{$\T$} (GGT);
\draw[-Stealth] (GGC) to node[left]{$\A$} (GCA);
\draw[-Stealth] (GGC) to node[below]{$\T$} (GCT);
\draw[-Stealth] (TGG) to node[left]{$\G$} (GGG);
\draw[-Stealth] (CGG) to node[above, sloped, xshift = -0.2em, yshift = -0.2em]{$\G$} (GGG);
\draw[-Stealth] (GGG) to node[above, sloped]{$\A$} (GGA);
\draw[-Stealth, bend right = 30] (GGG) to node[left]{$\C$} (GGC);
\draw[-Stealth] (GGG) to node[below, sloped, yshift = 0.2em] {$\T$} (GGT);

\draw[-Stealth] (GCA) to node[right]{$\A$} (CAA);
\draw[-Stealth] (GCA) to node[above, sloped]{$\T$} (CAT);
\draw[-Stealth] (GCA) to node[above, sloped]{$\G$} (CAG);
\draw[-Stealth] (CAG) to node[right]{$\C$} (AGC);
\draw[-Stealth] (AGC) to node[above]{$\A$} (GCA);
\draw[-Stealth] (AGC) to node[right]{$\T$} (GCT);
\draw[-Stealth] (CGG) to node[above, sloped, xshift = -0.75em, yshift = -0.2em]{$\C$} (GGC);
\draw[-Stealth] (CGG) to node[right]{$\A$} (GGA);
\draw[-Stealth] (GGA) to node[below]{$\C$} (GAC);
\draw[-Stealth] (CGG) to node[above, sloped, xshift =0.2em]{$\T$} (GGT);
\draw[-Stealth] (ACG) to node[above]{$\G$} (CGG);
\draw[-Stealth] (GAC) to node[right]{$\G$} (ACG);
\draw[-Stealth] (ACG) to node[above]{$\T$} (CGT);
\draw[-Stealth] (CGT) to node[above, sloped]{$\T$} (GTT);
\draw[-Stealth] (AGT) to node[right]{$\T$} (GTT);
\draw[-Stealth] (AGT) to node[below, sloped]{$\C$} (GTC);
\draw[-Stealth] (CAG) to node[below, sloped]{$\T$} (AGT);
\draw[-Stealth] (GTT) to node[above, sloped]{$\C$} (TTC);
\draw[-Stealth] (GGT) to node[above, sloped, xshift = 0.5em, yshift = -0.1em]{$\T$} (GTT);
\draw[-Stealth] (GGT) to node[left]{$\C$} (GTC);
\draw[-Stealth] (TCG) to node[left]{$\T$} (CGT);
\draw[-Stealth] (TCG) to node[above, sloped, xshift = 1.5em]{$\G$} (CGG);
\draw[-Stealth] (GTC) to node[below, sloped, xshift = -0.2em]{$\G$} (TCG);
\draw[-Stealth] (GTC) to node[above]{$\C$} (TCC);
\draw[-Stealth] (TTC) to node[below, sloped, xshift = 0.3em]{$\G$} (TCG);
\draw[-Stealth] (TTC) to node[above]{$\C$} (TCC);
\draw[-Stealth] (CGT) to node[above, sloped, xshift = -0.5em]{$\C$} (GTC);
\end{tikzpicture}   
\hfill
\begin{tikzpicture}[xscale=0.8]
\node[mynode] (TG) at (0,0) {{\T\G}};
\node[mynode] (GG) at (0,-1) {\G\G};
\node[mynode] (GC) at (0,-2) {\G\C};
\node[mynode] (CA) at (0,-3) {\C\A};
\node[mynode] (AA) at (1.75,-4) {\A\A};
\node[mynode] (AT) at (0,-4) {\A\T};
\node[mynode] (GA) at (2,0) {\G\A};
\node[mynode] (CG) at (2,-1) {\C\G};
\node[mynode] (GT) at (3.5,-2) {\G\T};
\node[mynode] (AC) at (3.5,0) {\A\C};
\node[mynode] (CT) at (2,-2) {\C\T};
\node[mynode] (AG) at (2,-3) {\A\G};
\node[mynode] (TC) at (3.5,-1) {\T\C};
\node[mynode] (CC) at (5.5,-1) {\C\C};
\node[mynode] (TT) at (5.5,-2) {\T\T};

\draw[-Stealth, label] (TG) to node[left]{$\G$} (GG);
\draw[-Stealth]  (GG) to node[left, above, sloped]{$\A$} (GA);
\draw[-Stealth] (GG) to node[below, sloped, xshift=-1em]{$\T$} (GT);
\draw[-Stealth] (GG) to node[left]{$\C$} (GC);
\draw[-Stealth] (GG) to [loop left] node[left] {\G} (GG);
\draw[-Stealth] (GC) to node[below, yshift= 0.2em]{$\T$} (CT);
\draw[-Stealth] (GC) to node[left]{$\A$} (CA);
\draw[-Stealth] (CA) to node[left]{$\A$} (AA);
\draw[-Stealth] (CA) to node[left]{$\T$} (AT);
\draw[-Stealth] (CA) to node[below]{$\G$}  (AG);
\draw[-Stealth] (GA) to node[above]{$\C$} (AC);
\draw[-Stealth] (CG) to node[above]{$\G$} (GG);
\draw[-Stealth] (AC) to node[above, sloped]{$\G$} (CG);
\draw[-Stealth] (CG) to node[above, sloped]{$\T$} (GT);
\draw[-Stealth] (GT) to node[below]{$\T$} (TT);
\draw[-Stealth] (GT) to node[right]{$\C$} (TC);
\draw[-Stealth] (AG) to node[below, sloped]{$\T$} (GT);
\draw[-Stealth] (AG) to node[below, sloped]{$\C$} (GC);
\draw[-Stealth] (TT) to node[above, sloped]{$\C$} (TC);
\draw[-Stealth] (TC) to node[above]{$\G$} (CG);
\draw[-Stealth] (TC) to node[above]{$\C$} (CC);
\end{tikzpicture}   
\caption{Node-centric (left) and edge-centric (right) deBruijn graphs for input string set $I = \{\T\G\G\A\C\G\G\G\A\C\G\G\C\A\T, \C\A\G\T\T\C\C, \C\G\G\T\C\G\T\T, \G\G\C\A\G\C\T\}$ and $k=3$. On the left, nodes correspond to $k$-mers, and we have edges connecting $k$-mers that have an overlap of $k-1$ (edge labels are omitted). On the right, the nodes are the $(k-1)$-mers of $I$, and $k$-mers are given by edges: edge $(u,v,c)$ represents $k$-mer $uc$. Note that the number of nodes of the graph on the left is equal to the number of edges of the graph on the right (both equal to 21, the number of distinct $k$-mers of $I$). }
\label{fig:dBG-node-edge}
\end{figure}
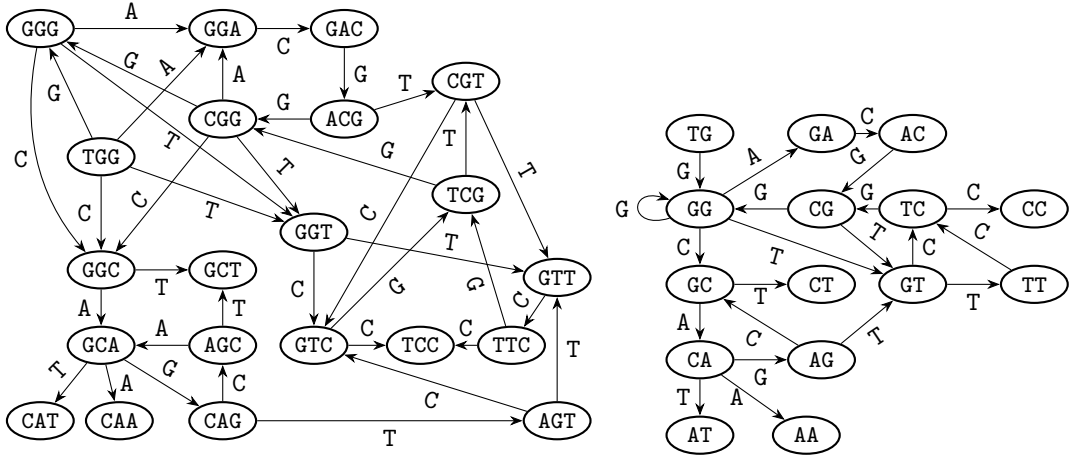

\gp{We can therefore build an SPSS by finding a} 
\textit{path \gp{(node)} cover} of the node-centric dBG, which is a collection of vertex-disjoint paths such that every node of the dBG belongs to exactly one path. Hence, all the $k$-mers in $I$ are represented only once in such a path cover. A minimum path cover minimizes the number of paths. Greedy algorithms like UST and its variants (UST-Compress, ESS-Compress, ESS-Tip-Compress) achieve near-optimal compression, outperforming traditional unitig-based and general-purpose compression methods by up to an order of magnitude \cite{rahman2021disk,rahman2021representation}; iterative SPSS decomposition and parallel algorithms further reduce storage and memory requirements~\cite{kitaya2021spss}. 
\gp{Equivalently, an SPSS can be found by looking for a \emph{trail (edge) cover} in the edge-centric dBG: a collection of edge-disjoint trails such that every edge of the dBG belongs to exactly one trail.}
In particular, the Eulertigs approach~\cite{schmidt2023eulertigs} constructs a \textit{minimum} \gp{such} cover of the edge-centric dBG \gp{(that is, minimum number of trails) by employing Eulerian tours}, \gp{providing} the shortest (i.e., smallest weight) SPSS representation in theory. \gp{Such representation only requires
$k-1$ further characters per trail, leading to a total size of $|\mathrm{spec}_k(I)| + (k-1)b$, where $b$ is the number of trails in the minimum cover}. This is known to yield very competitive no-repetition SPSS in practice and has become a strong baseline in graph-based compressors and indexers. 

\smallskip \noindent \textbf{Our contributions.}
We address a variant of the problem of identifying Spectrum-Preserving String Sets (SPSS) without repetitions (simplitigs) by exploiting the presence of cycles and tree-like attachments. While much of the literature reasons in terms of paths (and unitigs), biological data often contains circular molecules (plasmids, organellar genomes, complete bacterial chromosomes) \gp{and repeats}, suggesting that cycles could be considered in the node cover. 
For example, the string set $I = \{\mathtt{GCTGCGA}, \mathtt{AGGTT}, \mathtt{GTA}, \mathtt{ATCAC}, \mathtt{CAATA}\}$ is already minimal for SPSS with $k=3$; that is, $I$ is its own SPSS, since each $k$-mer in $I$ appears exactly once. However, $\mathtt{GCTGCGA}$ contains the \emph{circular} substring $\mathtt{GCTGC}$, where the first $k-1$ symbols $\mathtt{GC}$ are equal to the last $k-1$ symbols, allowing the latter to be omitted.
\gp{Furthermore, if we allow branching structures instead of simple paths, carefully handling the branching points, we could further reduce the representation size by avoiding the repetition of the first $k-1$ characters of the attaching $k$-mers.}

\gp{Based on these intuitions,} we extend the notion of path covers to account for branching structures that naturally arise in de~Bruijn graphs, where topologies are rarely limited to simple paths and cycles. \gp{Namely, we identify a (node) cover of what we call \emph{necklaces}\footnote{Also known as outerplanar, or functional, subgraphs}, which are subgraphs where each node has in-degree at most one. Intuitively, necklaces are formed by either a base cycle or path (serving as the ``necklace string'') with attached arborescences, called \emph{pendants} (serving as the ``necklace beads'') (see Figure~\ref{fig:treenecklaces}). Note that when the base is a path, the necklace is a tree. A \emph{necklace cover} is simply a node cover of the node-centric dBG in which every node belongs to exactly one (open or closed) necklace. Note that paths and cycles are special cases of necklaces without pendants. To {capture} 
such structures {and compactly handle the branching points}, we introduce a \emph{parenthesis representation} {for necklace covers.} In this way, we can define the cost of a necklace cover as the cumulative length of its parenthesis representation, and a \emph{minimum necklace cover} as one minimizing such cost. 
Since strings could need to be concatenated for storage or streaming, we can also consider a 
\emph{separator-based} representation of a set $S$ of strings: we append to each string of $S$ (except the last) the special character $|$, and then we concatenate all of these to form a single string over $\Sigma \cup \{|\}$. Since $|S|-1$ further separator characters were added to yield this representation, we can define the
\emph{separator-aware cost} as $\mathrm{cost}(S)=w(S)+|S|-1$. For example, for the previously considered SPSS $S = \{\texttt{CAGTTCC}, \texttt{TGGCT}, \texttt{CAA}, \texttt{AGCAT}, \texttt{GGGTCGGACGT}\}$, we obtain the separator-based representation $\texttt{CAGTTCC} |\texttt{TGGCT}| \texttt{CAA}| \texttt{AGCAT}| \texttt{GGGTCGGACGT}$, of size $35 = w(S) +|S|-1$.
}

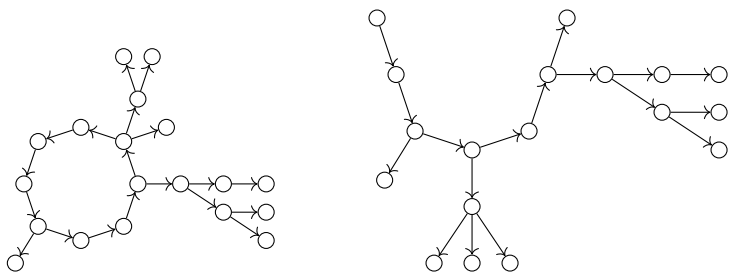
\begin{figure}[htbp]
\centering
\begin{tikzpicture}[scale=0.75]
    \vertex (1) at (0,0) {};
    \vertex (2) at (0.75,0.25) {};
    \vertex (3) at (1,1) {};
    \vertex (4) at (0.75,1.75) {};
    \vertex (5) at (0,2) {};
    \vertex (8) at (-0.75,0.25) {};
    \vertex (7) at (-1,1) {};
    \vertex (6) at (-0.75,1.75) {};
    \vertex (9) at (1.75,1) {};
    \vertex (10) at (1.5, 2) {};
    \vertex (10bis) at (1, 2.5) {};
    \vertex (11) at (-1.15, -0.4) {};
    \vertex (12) at (2.5, 1) {};
    \vertex (12bis) at (3.25, 1) {};
    \vertex (13) at (1.25, 3.25) {};
    \vertex (13bis) at (0.75, 3.25) {};

    \vertex (14) at (2.5,0.5) {};
    \vertex (15) at (3.25,0.5) {};
    \vertex (16) at (3.25,0) {};

    \path[->]
    (1) edge (2)
    (2) edge (3)
    (3) edge (4)
    (4) edge (5)
    (5) edge (6)
    (6) edge (7)
    (7) edge (8)
    (8) edge (1)
    (3) edge (9)
    (9) edge (12)
    (4) edge (10)
    (4) edge (10bis)
    (8) edge (11)
    (10bis) edge (13)
    (12) edge (12bis)
    (10bis) edge (13bis)
    (9) edge (14)
    (14) edge (15)
    (14) edge (16)
    ;
\end{tikzpicture}
\hspace{1cm}
\begin{tikzpicture}
    \vertex (1) at (0,0) {};
    \vertex (2) at (0.75,0.25) {};
    \vertex (3) at (1,1) {};
    \vertex (4) at (1.25, 1.75) {};
    \vertex (6) at (-1.25,1.75) {};
    \vertex (7) at (-1,1) {};
    \vertex (8) at (-0.75,0.25) {};
    \vertex (9) at (1.75,1) {};
    \vertex (10) at (0,-0.75) {};
    \vertex (10l) at (-0.5,-1.5) {};
    \vertex (10r) at (0.5,-1.5) {};
    \vertex (11) at (-1.15, -0.4) {};
    \vertex (12) at (0, -1.5) {};
    \vertex (13) at (2.5,1) {};
    \vertex (14) at (3.25,1) {};
    
    \vertex (15) at (2.5,0.5) {};
    \vertex (16) at (3.25,0.5) {};
    \vertex (17) at (3.25,0) {};

    \path[->]
    (1) edge (2)
    (2) edge (3)
    (3) edge (4)
    (6) edge (7)
    (7) edge (8)
    (8) edge (1)
    (3) edge (9)
    (8) edge (11)
    (1) edge (10)
    (10) edge (12)
    (10) edge (10l)
    (10) edge (10r)
    (9) edge (13)
    (13) edge (14)
    (9) edge (15)
    (15) edge (16)
    (15) edge (17)
    ;
\end{tikzpicture}
\caption{Closed (left) and open (right) necklaces, with resp.~4 and~3 pendants.}
\label{fig:treenecklaces}
\end{figure}

We present a linear-time greedy algorithm \gp{ \greedyNeckAlgo for finding a necklace cover of a given node-centric dBG. The algorithm takes as input any node cover composed of paths and cycles (a \emph{PC cover}), and from this it} constructs a {necklace cover}. \gp{We prove that
our algorithm outputs a minimum necklace cover (i.e., minimum cumulative length of its parenthesis representation).}
\gp{Since the size of a minimum necklace cover} is always \gp{smaller} than (or equal to) the size of the minimum SPSS (the latter is a special case of the former \gp{as paths are trivial necklaces)}, 
\gp{our algorithm produces a representation of $\mathrm{spec}_k(I)$ that is always smaller than the minimum SPSS. On the other hand, a PC cover is easy to obtain (greedily only keep one in-neighbor and out-neighbor for all nodes).}  
Furthermore, we prove that there exists an infinite family of input sets $I$ for which \gp{our representation is always better than Eulertigs: the parenthesis representation of the minimum necklace cover} requires only a fraction $4/(k+1)$ of the symbols needed by the minimum SPSS.

\gp{We also perform an experimental evaluation of our algorithm. We compare our proposed method against both a fully-greedy baseline for finding necklace covers, and state-of-the-art SPSS methods, namely}
Eulertigs~\cite{schmidt2023eulertigs} and Masked Superstring~\cite{sladky2023masked},
measuring the size of the resulting representation \gp{(as its cumulative size)} and execution time for $k$-mers in several biological datasets.
Our preliminary results show that the smallest space occupancy for small values of $k$ is obtained by the Masked Superstring method, while our proposed methods using the pseudo-forest representation of necklaces achieved comparable performance. For larger values of $k$, our methods attained the smallest space occupancy among all tested algorithms.

It should be noted, however, that the Masked Superstring method does not solve the SPSS problem exactly, as it may introduce false positives, i.e., $k$-mers that do not occur in the original input sequences. In contrast, our proposed method, as well as Eulertigs, 
guarantee an exact $k$-mer spectrum.  
The main insight from these experiments is that by using necklaces produced by our greedy algorithm
we achieve a space reduction comparable to or better than that of the superstring-based method Masked Superstring, but without introducing false positives in the SPSS representation. This confirms that the use of necklace structures, combining cycles and open paths with pendants to obtain a parenthesis representation, provides a favorable trade-off between space efficiency and computational cost, demonstrating that circularity and pendant trees in necklaces are key factors for achieving compact representations.

\subsection{Notation and Preliminaries}
\label{section:preliminaries}

\paragraph*{Strings} A \emph{string} $S= S[1]\cdots S[|S|]$, of length $|S|$, is a sequence of characters $S[i]$ from an alphabet $\Sigma$. $T$ is a \emph{substring} of $S$ if $T = S[i]\cdots S[i+|T|]$ for some $i \le |S|$. The \emph{prefix} (resp. \emph{suffix}) of $S$ of length $\ell$, denoted $pre_\ell(S)$ (resp. $suf_\ell(S)$), is the string $S[1]\cdots S[\ell]$ (resp. $S[|S| - \ell +1]\cdots S[|S|]$). A \emph{$k$-mer} of $S$ is a substring of length $k$. The \emph{spectrum} of a set of strings $I$, denoted $\mathrm{spec}_k(I)$ is the set of all $k$-mers of the strings of $I$. 
The \emph{spectrum preserving string set} (SPSS) of $I$ is a set of strings $S$ such that $\mathrm{spec}_k(S) = \mathrm{spec}_k(I)$. In this work, unless otherwise stated, we will consider SPSS as being \emph{without repetitions}, that is, we also assume that each $k$-mer of $\mathrm{spec}_k(I)$ appears exatcly once in $S$. 
The \emph{weight} of a set of strings is its total character count: $w(I) = \sum_{s\in I} |s|$. We consider the \emph{separator-based representation} for a set $S$ of strings, where we concatenate the strings of $S$, separating them with a special character $|$. For this representation, we consequently define the \emph{separator-aware cost} as the length of the resulting string, i.e. $\textrm{cost}(S) = w(S) + |S| -1$.
A \emph{circular string} is a string $X$ where we assume that the last character of the string connects back to the first, i.e., $X[|X|+i] = X[i]$ for all $i$. For instance, if $X= \A\C\G\G\T$ is circular, then its $k$-mers (for $k=3$) are not only $\A\C\G, \C\G\G,$ and $\G\G\T$, but $\G\T\A$ ($\underline{\A}\C\G\underline{\G\T}$) and $\T\A\C$ ($\underline{\A\C}\G\G\underline{\T}$) as well.

\paragraph*{Graphs} A directed graph is a pair $G=(V(G),E(G))$, where $V(G)$ is the set of nodes, and $E(G) \subseteq V(G)\times V(G)$ is its set of edges.
In a \emph{labeled graph}, each edge $(u,v)$ has an associated label $c$; we denote such a labeled edge as $(u,v,c)$. A \emph{subgraph} $H$ of $G$ is a graph such that $V(H) \subseteq V(G)$ and $E(H) \subseteq E(G)$. 
For a node $v$, its \emph{out-neighbors}, or \emph{adjacent nodes}, are the nodes $N^+(v) = \{w \in V(G) \ | \ (v,w) \in E(G)\}$. The number of out-neighbors of a node is called its \emph{out-degree} and is denoted as $d^+(v)$. Symmetrically, we define its \emph{in-neighbors} as $N^-(v) = \{u \in V(G) \ | \ (u,v) \in E(G)\}$, and its cardinality as the \emph{in-degree} of the node, denoted $d^-(v)$. 
A \emph{path} of length $k$ is a sequence of $k$ distinct adjacent nodes: $u_1,u_2,\ldots,u_k$, such that $(u_i, u_{i+1})\in E(G)$ for all $i<k$, and $u_i \neq u_j$ for all $i\neq j$. 
A \emph{trail} of length $k$ is a sequence of $k$ distinct adjacent edges: $(u_1,u_2), (u_2,u_3), \ldots (u_k,u_{k+1})$, such that the $(u_i, u_{i+1})$ are distinct for all $i\le k$ (but nodes may repeat). A \emph{node cover} (resp. \emph{edge cover)} of graph $G$ is a set of subgraphs $H_1,...,H_k$ of $G$ such that $\cup_i V(H_i) = V(G)$ (resp. $\cup_i E(H_i) = E(G)$) and $V(H_i) \cap V(H_j) = \emptyset$ (resp. $E(H_i) \cap E(H_j) = \emptyset$). More specifically, $H_1,...,H_k$ is a \emph{path node cover} if each $H_i$ is a path, a \emph{path-and-cycle cover} (PC cover) if each $H_i$ is either a path or a cycle, and a \emph{trail edge cover} if each $H_i$ is a trail. 
 
A \emph{necklace} of $G$ is a connected subgraph where each node has in-degree at most 1. Note that this implies that a necklace is formed by either a cycle or a path, called the \emph{root} of the necklace, with attached arborescences (that is, directed trees), called \emph{pendants}. When the root is a cycle we have a \emph{closed} necklace, otherwise the whole necklace is an arborescence, and we also refer to it as an \emph{open} necklace. 
See Figure~\ref{fig:treenecklaces} for examples. 

\paragraph*{deBruijn Graphs} deBruijn graphs are used to model relationships between $k$-mers of a string, or set of strings. The \emph{order-$k$ node-centric deBruijn graph} for an input string set $I$ is a graph $G=(V,E)$ where each node is a $k$-mer of $I$ ($V = \mathrm{spec}_k(I)$), and we have an edge $(u,v)\in E$ if and only if the suffix of length $k-1$ of $u$ is a prefix of $v$, i.e. $suf_{k-1}(u) = pre_{k-1}(v)$. We can equip each edge of $G$ with labels in a natural way: the label of edge $(u,v)$ is given by the last character $c$ of $v$: in this way, $v = suff_{k-1}(u) c$. See the left of Figure~\ref{fig:dBG-node-edge} for an example. On the other hand, the \emph{order-$k$ edge-centric deBruijn graph} for $I$ is a labeled graph $G' = (V',E')$ where each node is a $(k-1)$-mer of $I$ ($V' = \mathrm{spec}_{k-1}(I)$), and $(u,v,c)\in E'$ if and only if the suffix of length $k-2$ of $u$ is a prefix of $v$, i.e. $suf_{k-2}(u) = pre_{k-2}(v)$, and the concatenation $uc$ is a $k$-mer of $I$, i.e. $uc\in \mathrm{spec}_k(I)$. See the right of Figure~\ref{fig:dBG-node-edge} for an example. Each path in $G$ corresponds to a trail in $G'$, and vice versa. From this it follows that is a 1:1 correspondence between path node covers of the node-centric dBG and trail edge covers of the edge-centric dBG.
It is easy to see that any order-$k$ edge-centric deBruijn graph is an order-$(k-1)$ node-centric deBruijn graph as well (albeit possibly for a different input set of strings, as not all edges between compatible $(k-1)$-mers are added).

\paragraph*{Eulerian Tours and Eulertigs}
We give here a quick background on Eulertigs~\cite{schmidt2023eulertigs}, the state-of-the-art algorithm for finding a minimum SPSS (i.e. of minimum weight). The idea is based on \emph{Eulerian tours}: an Eulerian tour of a graph $G$ is a trail that visits each edge in $E(G)$ exactly once, and starts and ends at the same node. It can be seen as a trail edge cover of $G$ formed by just one trail. Graphs that admit an Eulerian tour are called \emph{Eulerian graphs}. 
A graph is Eulerian if and only if for each $v\in V(G)$, its outdegree is equal to its indegree: $d^+(v)= d^-(v)$. A node for which these values are different is called \emph{unbalanced}. 
Given an Eulerian graph, it then takes linear time to find an Eulerian tour. 

The idea behind Eulertigs is as follows. Start from the edge-centric dBG $G$ of $I$, and add the minimum number of edges necessary to make the graph Eulerian, yielding $G'$. These are fake unlabeled edges (called \emph{breaking edges}) that will later be disregarded; as such, they do not respect the $k$-mer overlap conditions. Note that, in this way, $G'$ is not technically a deBruijn graph anymore. Then, look for an Eulerian tour of $G'$, starting with a breaking edge. Build the corresponding SPSS by spelling the strings along the edges of the Eulerian tour, beginning a new string anytime a breaking arc is traversed. The authors prove that this yields a minimum SPSS. Note that, if there are $b$ breaking edges, the weight of the final representation is $|\mathrm{spec}_k(I)| + (k-1)b$: as every time we start a new string we need to explicitly spell the $(k-1)$-prefix of the first $k$-mer, while when extending an existing string we are just adding one character per $k$-mer. 

While making a deBruijn graph Eulerian in a way that preserves the deBruijn graph's conditions is NP-hard~\cite{bernardini2022making}, making a regular graph Eulerian requires linear time. Indeed, it is sufficient to fix every pair of oppositely-unbalanced nodes: find a pair of nodes $u,v$ such that $d^+(u) - d^-(u) <0$ and $d^+(v) - d^-(v) >0$, and add edge $(u,v)$; repeat. Note that, by the handshaking lemma, such a pair exists if and only if the graph is not Eulerian. Thus, the set of breaking edges with minimum cardinality can be found in linear time, leading to a linear time algorithm overall. 

As an example, consider the edge-centric deBruijn graph on the right of Figure~\ref{fig:dBG-node-edge}, corresponding to input string set $I = \{\T\G\G\A\C\G\G\A\C\G\G, \C\A\G\T\T\C\C, \C\G\G\T\C\G\T\T\A, \G\G\C\A\G\C\T\}$. The minimum number of breaking edges to make it Eulerian is 5 (there are ten unbalanced nodes: \T\G, \C\C, \G\G, \C\T, \G\T, \A\G, \C\A, \T\A, \A\A, and \T\T), and a possible Eulerian tour is given by (breaking edges unlabeled) 
\texttt{CA \dBGedge{G} AG \dBGedge{T} GT \dBGedge{T} TT \dBGedge{C} TC \dBGedge{C} CC \dBGedge{} TG \dBGedge{G} GG \dBGedge{C} GC \dBGedge{T} CT \dBGedge{} CA \dBGedge{A} AA \dBGedge{} AG \dBGedge{C} GC \dBGedge{A} CA \dBGedge{T} AT \dBGedge{} GG \dBGedge{G} GG \dBGedge{T} GT \dBGedge{C} TC \dBGedge{G} CG \dBGedge{G} GG \dBGedge{A} GA \dBGedge{C} AC \dBGedge{G} CG \dBGedge{T} GT}, corresponding to minimum SPSS 
$S = \{\texttt{CAGTTCC}, \texttt{TGGCT}, \texttt{CAA}, \texttt{AGCAT}, \texttt{GGGTCGGACGT}\}$
of weight $31 = 21 + 2\times5 = |\mathrm{spec}_k(I)| + (k-1)b$. 

\subsection{Roadmap}
We start by presenting our parenthesis representation of necklaces in Section~\ref{section:parenthesis-representation}. This representation will give us the metrics needed to define a minimum necklace cover. Then, in Section~\ref{section:tree-necklace-algo}, we can present our algorithm \greedyNeckAlgo, to find such a minimum necklace cover. In that section, we also prove that there is an infinite family of graphs for which the minimum necklace cover parenthesis representation is smaller than the SPSS. Finally, Section~\ref{section:experiments} presents the experimental evaluation  of necklace covers over real biological datasets, comparing it with state-of-the-art SPSS methods.

\section{Representing Necklace Covers}
\label{section:parenthesis-representation}
In this section we present a balanced-parenthesis representation of necklace covers, with optional separator-based representation as well. Intuitively, we show how to represent a necklace through a string with added (balanced) parenthesis representing the branches. 

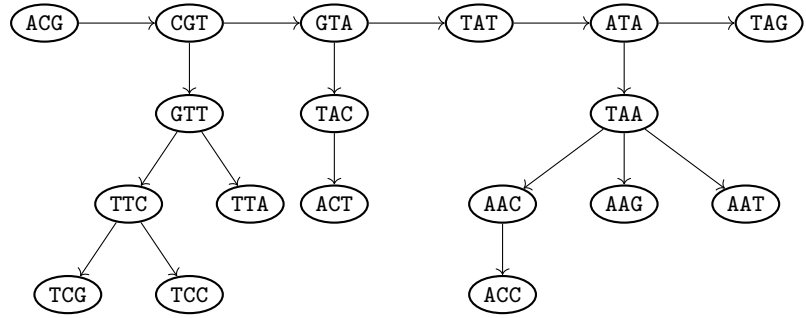
\begin{figure}
    \centering
\begin{tikzpicture}[
  every node/.style={mynode, draw},
  edge from parent/.style={draw, ->},
  level distance=12mm,
  sibling distance=16mm
]


\node[mynode] (x2) {$\C\G\T$}
  child { node (b) {$\G\T\T$}
    child { node (e) {$\T\T\C$}
      child { node {$\T\C\G$} }
      child { node {$\T\C\C$} }
    }
    child { node {$\T\T\A$} }
  };

\node[mynode, left=10mm of x2] (x1) {$\A\C\G$};
\node[mynode, right=10mm of x2] (x3) {$\G\T\A$}   
child { node (c) {$\T\A\C$}
    child { node {$\A\C\T$} }
  };
\node[mynode, right=10mm of x3] (x4) {$\T\A\T$};

\node[mynode, right=10mm of x4] (x5) {$\A\T\A$}
  child { node {$\T\A\A$}
    child { node {$\A\A\C$}
      child { node {$\A\C\C$} }
    }
    child { node {$\A\A\G$}}
    child { node {$\A\A\T$}}
  };

\node[right=10mm of x5] (x6) {$\T\A\G$};

\draw[->] (x1) -- (x2);
\draw[->] (x2) -- (x3);
\draw[->] (x3) -- (x4);
\draw[->] (x4) -- (x5);
\draw[->] (x5) -- (x6);

\end{tikzpicture}

\caption{
Bottom: open necklace with same pendant structure as the subtrees of the root on the top, with parenthesis representation given by $\texttt{ACGT(T(C(G)C)A)A(CT)TA(A(CC)(G)T)G}$ }
\label{fig:BP-tree}
\end{figure}

In the usual way to represent trees through balanced parentheses (BP)~\cite{jacobsonBP63533,bp}, we have a set of parentheses for each node, enclosing the parentheses of its children. 
This can be built by traversing the tree in preorder, opening a parenthesis when visiting a node for the first time, and closing a parenthesis when going back through it, after having visited its subtree.
In our case, we want to embed a similar nested parentheses structure, modeling the branchings inside the node pendants, inside of a string spelling the $k$-mers of the root necklace. Thus, after adding the root path (or cycle) as a string, we embed its pendants inside pairs of parentheses, at the positions of the string corresponding to the $k$-mers where the corresponding branching occurs. Each pendant is encoded in tree BP representation, where the edge labels are kept: every time we visit a node for the first time, entering it with label $\ell$, we add $(\ell$ to the string. We then close the parentheses as usual after visiting all of the node's subtree.

For instance, for the necklace on the bottom of Figure~\ref{fig:BP-tree}, we start with $\texttt{ACGTATAG}$. Then, add a set of parenthesis for each tree stemming from such root path, at the position in the string corresponding to the $k$-mer where the branching occurs:
$\texttt{ACGT()A()TA()G}$. We build the BP representation for its three subtrees, respectively given by $\texttt{T(C(G)(C))(A)}$, $\texttt{C(T)}$, and $\texttt{A(C(C))(G)(T)}$. We note that we can further reduce the number of parenthesis without losing information, by omitting the enclosing parentheses for the last child of a node. Indeed, after the previous children, all that is left before the closing parenthesis of the parent must correspond to the subtree of the last child. 
This also implies that, when we have unary paths (as in the second and third subtree here), we just concatenate the corresponding labels without using parentheses. 
In this setting, the three subtrees of before would become $\texttt{T(C(G)C)A}$, $\texttt{CT}$, and $\texttt{A(CC)(G)T}$, removing 10 further parentheses, and leading to the final representation: 
$\texttt{ACGT(T(C(G)C)A)A(CT)TA(A(CC)(G)T)G}$.

In the resulting representation, each pair of parentheses represents a branch, or a choice between characters with which we can extend the current string. Thus, each pair of parenthesis corresponds to a leaf of the necklace. 
Unary paths (without branches) are encoded without parentheses, as in the case of the central subtree of the bottom of Figure~\ref{fig:BP-tree}, simply expressed as \C\T.
Enclosed parenthesis represent branches at different levels, while $c$ consecutive parentheses represent a branch of $c+1$ children for the previous node. Indeed, in string $\texttt{T(C(G)C)A}$ (corresponding to the leftmost subtree in the bottom of Figure~\ref{fig:BP-tree}) the first set of parentheses expresses the possibility of branching to the subtree rooted in $\texttt{TTC}$, or proceeding to $\texttt{A}$, while the internal set of parenthesis gives us the next-level branch, choosing between branching to $\texttt{TCG}$, or continuing to $\texttt{TCC}$. 
Instead, in $\texttt{A(CC)(G)T}$, representing the rightmost subtree in the bottom of Figure~\ref{fig:BP-tree}, the two consecutive pairs of parentheses represent the possibility of branching to $\texttt{CC}$ (left), or to $\G$ (center), or lastly $\T$ (right).

If the root of the necklace is a cycle, the only difference is that we start from the corresponding circular string (omitting the last $k-1$ characters of the string) instead of a regular string, and then we proceed exactly as above.

\begin{figure}
    \centering
\begin{tikzpicture}[xscale=0.8]
\node[mynode, fill = red!60!white] (TGG) at (-0,0) {{\T\G}\G};
\node[mynode, fill = red!60!white] (GGG) at (-1,1.7) {{\G\G}\G};
\node[mynode, fill = green!70!black] (GGC) at (0,-1.5) {\G\G\C};
\node[mynode, fill = green!70!black] (GCA) at (0,-2.5) {\G\C\A};
\node[mynode, fill = green!70!black] (CAA) at (0.3,-3.5) {\C\A\A};
\node[mynode, fill = green!70!black] (CAT) at (-1,-3.5) {\C\A\T};
\node[mynode, fill = green!70!black] (GGA) at (2,1.7) {\G\G\A};
\node[mynode, fill = green!70!black] (CGG) at (2,0.5) {\C\G\G};
\node[mynode, fill = red!60!white] (GGT) at (3.5,-1) {\G\G\T};
\node[mynode, fill = green!70!black] (GAC) at (4,1.7) {\G\A\C};
\node[mynode, fill = green!70!black] (ACG) at (4,0.5) {\A\C\G};
\node[mynode, fill = green!70!black] (AGC) at (2,-2.5) {\A\G\C};
\node[mynode, fill = green!70!black] (GCT) at (2,-1.5) {\G\C\T};
\node[mynode, fill = green!70!black] (CAG) at (2,-3.5) {\C\A\G};
\node[mynode, fill = blue!50!white] (GTC) at (3.5,-2.5) {\G\T\C};
\node[mynode, fill = blue!50!white] (TCC) at (5.1,-2.5) {\T\C\C};
\node[mynode, fill = blue!50!white] (TCG) at (6,-0.5) {\T\C\G};
\node[mynode, fill = blue!50!white] (CGT) at (6,1) {\C\G\T};
\node[mynode, fill = blue!50!white] (GTT) at (7.5,-1.6) {\G\T\T};
\node[mynode, fill = blue!50!white] (TTC) at (6.75,-2.5) {\T\T\C};
\node[mynode, fill = green!70!black] (AGT) at (7.5,-3.5) {\A\G\T};

\draw[-Stealth, label] (TGG) to node[left]{$\C$} (GGC);
\draw[-Stealth] (TGG) to node[above, sloped, xshift = 0.75em]{$\A$} (GGA);
\draw[-Stealth, very thick, red!60!white] (TGG) to node[below, sloped, xshift = 0.3em, yshift = 0.1em]{$\T$} (GGT);
\draw[-Stealth, very thick, green!70!black] (GGC) to node[left]{$\A$} (GCA);
\draw[-Stealth, very thick, green!70!black] (GGC) to node[below]{$\T$} (GCT);
\draw[-Stealth, very thick, red!60!white] (TGG) to node[left]{$\G$} (GGG);
\draw[-Stealth] (CGG) to node[above, sloped, xshift = -0.2em, yshift = -0.2em]{$\G$} (GGG);
\draw[-Stealth] (GGG) to node[above, sloped]{$\A$} (GGA);
\draw[-Stealth, bend right = 30] (GGG) to node[left]{$\C$} (GGC);
\draw[-Stealth] (GGG) to node[below, sloped, yshift = 0.2em] {$\T$} (GGT);

\draw[-Stealth, very thick, green!70!black] (GCA) to node[right]{$\A$} (CAA);
\draw[-Stealth, very thick, green!70!black] (GCA) to node[above, sloped]{$\T$} (CAT);
\draw[-Stealth, very thick, green!70!black] (GCA) to node[above, sloped]{$\G$} (CAG);
\draw[-Stealth, very thick, green!70!black] (CAG) to node[right]{$\C$} (AGC);
\draw[-Stealth] (AGC) to node[above]{$\A$} (GCA);
\draw[-Stealth] (AGC) to node[right]{$\T$} (GCT);
\draw[-Stealth, very thick, green!70!black] (CGG) to node[above, sloped, xshift = -0.75em, yshift = -0.2em]{$\C$} (GGC);
\draw[-Stealth, very thick, green!70!black] (CGG) to node[right]{$\A$} (GGA);
\draw[-Stealth, very thick, green!70!black] (GGA) to node[below]{$\C$} (GAC);
\draw[-Stealth] (CGG) to node[above, sloped, xshift =0.2em]{$\T$} (GGT);
\draw[-Stealth, very thick, green!70!black] (ACG) to node[above]{$\G$} (CGG);
\draw[-Stealth, very thick, green!70!black] (GAC) to node[right]{$\G$} (ACG);
\draw[-Stealth] (ACG) to node[above]{$\T$} (CGT);
\draw[-Stealth, very thick, blue!70!white] (CGT) to node[above, sloped]{$\T$} (GTT);
\draw[-Stealth] (AGT) to node[right]{$\T$} (GTT);
\draw[-Stealth] (AGT) to node[below, sloped]{$\C$} (GTC);
\draw[-Stealth, very thick, green!70!black] (CAG) to node[below, sloped]{$\T$} (AGT);
\draw[-Stealth, very thick, blue!70!white] (GTT) to node[above, sloped]{$\C$} (TTC);
\draw[-Stealth] (GGT) to node[above, sloped, xshift = 0.5em, yshift = -0.1em]{$\T$} (GTT);
\draw[-Stealth] (GGT) to node[left]{$\C$} (GTC);
\draw[-Stealth, very thick, blue!70!white] (TCG) to node[left]{$\T$} (CGT);
\draw[-Stealth] (TCG) to node[above, sloped, xshift = 1.5em]{$\G$} (CGG);
\draw[-Stealth] (GTC) to node[below, sloped, xshift = -0.2em]{$\G$} (TCG);
\draw[-Stealth] (GTC) to node[above]{$\C$} (TCC);
\draw[-Stealth, very thick, blue!70!white] (TTC) to node[below, sloped, xshift = 0.3em]{$\G$} (TCG);
\draw[-Stealth, very thick, blue!70!white] (TTC) to node[above]{$\C$} (TCC);
\draw[-Stealth, very thick, blue!70!white] (CGT) to node[above, sloped, xshift = -0.5em]{$\C$} (GTC);
\end{tikzpicture}   
\caption{Necklace cover for the graph of Figure~\ref{fig:dBG-node-edge}, formed by two closed necklaces (green and blue) and one open necklace (red).}
\label{fig:necklace-cover-example}
\end{figure}
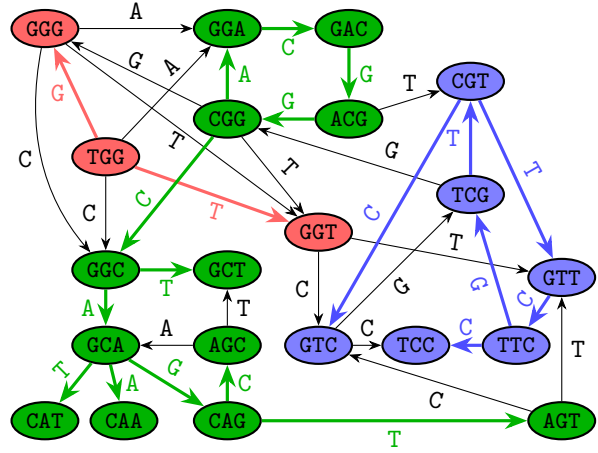

For example, consider the necklace cover shown in Figure~\ref{fig:necklace-cover-example}, for the same input set $I = \{\T\G\G\A\C\G\G\G\A\C\G\G\C\A\T, \C\A\G\T\T\C\C, \C\G\G\T\C\G\T\T, \G\G\C\A\G\C\T\}$ as in Figure~\ref{fig:dBG-node-edge} ($k=3$). This cover is formed by two closed necklaces with root cycles \texttt{GGAC} and \texttt{TCGT}, shown in green and blue respectively, and one open necklace, in red. The BP representations of the two closed necklaces are given by \texttt{GG(C(A(A)(T)G(T)C)T)AC} and \texttt{TC(C)GT(C)}, while the open necklace can be represented as \texttt{TGG(G)T}. On top of the necklaces, we need to retain information about which necklaces were open and closed.
Assuming that the strings of the cover are ordered, it suffices to put a single special character $\$$ to separate closed necklaces, which we place first, from open ones. 
Thus, the necklace cover is given by $\mathcal{C} = \{ \texttt{GG(C(A(A)(T)G(T)C)T)AC}, \texttt{TC(C)GT(C)}, \$ \texttt{TC(C)GT(C)}\}$.

To reconstruct the $i$-th $k$-mer, it suffices to proceed to the $k$-th character of $\Sigma$ in the string, and retrieve its parent $k-1$ times, concatenating the found letters. We have two cases for the parent, according to whether we are in a unary path: if the previous character is in $\Sigma$, then they are the parent, otherwise, if the previous character is $($, we are at the root of a subtree and we need to find the enclosing pair of parentheses to retrieve the parent. Note that we have two choices: either the character before $($ is in $\Sigma$, and we have found the parent, or it is $)$, closing the BP representation of the previous sibling. Thus, to jump to the parent, we need to jump to the corresponding opening parenthesis of the closing one right before our position, and repeat until we reach a letter of $\Sigma$. This can be done in constant time if we allow further linear space, using operation $findopen(i)$ \cite{munro2001succinct}. Since we do not want to add space, and we still need to linearly scan until position $i$, we do so by scanning backwards.
If we are at the outermost layer, and the string is circular, we further have to wrap around. 

We obtain the following:
\begin{lemma}
    \label{lemma:BP-representation-size}
    The parenthesis representation for a necklace cover of an input set $I$ can be computed in $O(w(I))$ time and space. Let $N_k$ be the number of distinct $k$-mers in $I$ Let $N_C$ be the number of closed necklaces, $N_O$ be the number of open necklaces, and $N_L$ be the number of leaves over all the pendants in the necklace cover. The resulting representation uses $N_k + (k-1) N_O + 2 N_L +1$ symbols from the alphabet $\Sigma \cup \{ \mathtt{(}, \$, \mathtt{)}\}$, where $2 N_L$ corresponds to the number of parentheses.
\end{lemma}
\begin{proof}
    Each $k$-mer in $I$ adds one extra symbol, which gives $N_k$. However, open necklaces needs extra $k-1$ symbols, and each leaf adds an extra pair of parenthesis, which gives $(k-1) N_O + 2 N_L$. We then need one extra \$ symbol to separate open and closed necklaces.
\end{proof}

We can therefore define the \emph{cost of a necklace cover} to be the length of its BP representation (as given in Lemma~\ref{lemma:BP-representation-size}), and consequently a \emph{minimum necklace cover} as a necklace cover of minimum cost. 

\paragraph*{Separator-based Necklace Representation}
We also explore a separator-based representation for a necklace cover, we concatenate the BP representations of necklaces, separating them with $|$ as usual. There is only one thing left to take care of: we need to identify which necklaces are closed, and which are open. Since the order of necklaces in the representation is irrelevant, we can do this by using only one extra separator: we place all of the closed necklaces first, and place an extra character $|$ when switching to open necklaces (so, at that point, we have two vertical bars in a row).
For the previous example, the final separator-based representation for the necklace cover is given by \texttt{GG(C(A(A)(T)G(T)C)T)AC|TC(C)GT(C)||TGG(G)T}, and has length 42. Note that this is not a minimum necklace cover (the graph can be covered by a spanning tree = a single open necklace).

\begin{lemma}
\label{lemma:BP-representation-size-separators}
The parenthesis representation for a necklace cover of an input set $I$ can be computed in $O(w(I))$ time and space.  Let $N_k$ be the number of distinct $k$-mers in $I$. Let $N_C$ be the number of closed necklaces, $N_O$ be the number of open necklaces, and $N_L$ be the number of leaves over all the pendants in the necklace cover. The resulting representation uses $N_k + k N_O + 2 N_L + N_C$ symbols from the alphabet $\Sigma \cup \{ \mathtt{(}, \mathtt{\mid}, \mathtt{)}\}$, where $2 N_L$ corresponds to the number of parentheses.
\end{lemma}
\begin{proof}
The analysis is as in Lemma~\ref{lemma:BP-representation-size}, with the only difference given by the addition of separators instead of the \$ symbol. We need vertical bars to separate individual necklaces, plus one extra vertical bar to separate the closed necklaces from the open ones, which gives $N_C + N_O$ characters.
\end{proof}

\section{Finding a Minimum Necklace Cover}
\label{section:tree-necklace-algo}
Now that we have defined the cost of a necklace cover, we present an algorithm that finds the minimum such cover in the separator-oblivious model. We will later discuss the implications on the the separator-based model.

Our algorithm takes as input a PC cover, that is, a cover formed by paths and cycles, and outputs a necklace cover by greedily building necklace, attaching paths to partial necklaces and closing new cycles whenever possible.

\subsection{Minimum Path-and-Cycle Covers}
\label{section:min-pc-cover}
\gp{Our algorithm requires a PC cover as input: according to the type of cover, the performance of the algorithm may vary.}   

\gp{When representing the paths as strings, cycles correspond to circular strings. On a deBruijn graph, this allows us to save $k-1$ characters, as they are given by the circularity, with respect to paths. Intuitively, this means that the metric we are interested in is simply the number of paths, and we can allow any number of cycles.
Therefore, we define the \emph{cost} of a PC cover as the number of its paths, and consequently a \emph{minimum} PC cover as one with the minimum \emph{number of paths}, regardless of their number of cycles.} 

\gp{Note that each node of in-degree zero, called a \emph{primitive node}, must be the first node of a path in any PC cover. This type of path is called \emph{primitive}. Thus, if there are $p$ primitive nodes, then any PC cover must contain $p$ primitive paths, and has cost at least $p$.
The converse is not true, in the sense that a minimum PC cover can have cost strictly greater than $p$. For instance, if the graph is formed by a node with two disjoint exiting paths, we only have one primitive node, but we need two paths to have a cover.} 

\gp{One way to find a minimum PC cover is by using a maximum bipartite matching algorithm. A bipartite directed graph $(V_1\cup V_2, E_B)$ is a graph where the node set can be partitioned into two disjoint parts $V_1\cup V_2$, and edges only connect different parts: $(u,v)\in E_B $ implies $u\in V_1$ and $v\in V_2$. A matching of a bipartite graph is a set of edges that do not share endpoints (i.e., for any two edges of the matching $(u_1,v_1),(u_2,v_2)$ we have $u_1 \neq u_2$ and $v_1 \neq v_2$. A maximum matching is a matching of maximum size.  
Given an order-$k$ node-centric de Bruijn graph (dBG) $G=(V,E)$, we build a} bipartite graph $G'=(V_L\cup V_R,E')$ by creating two disjoint copies of the vertex set, denoted $V_L$ and $V_R$, together with bijections from the original vertex set to the copies, $f_L:V\to V_L$ and $f_R:V\to V_R$. For each directed edge $(u,v)\in E$, the bipartite graph contains an edge $(f_L(u),f_R(v))$, so that adjacency in $G$ is preserved across the bipartite graph.  
The next step is to compute a maximum bipartite matching on $G'$ \gp{(this can be done for instance in $\tilde{O}(|E| + |V|^{1.5})$ total running time using~\cite{BrandLNPSS0W20})}. Let $M'\subseteq E'$ denote the resulting matching; each edge in $M'$ is mapped back to the original edge set of $G$, thus producing a subset $F\subseteq E$. \gp{Let $G_F=(V,F)$ be the subgraph induced by the edges of the matching. Let $(V_1,E(V_1)),\ldots,(V_m,E(V_m))$ be its strongly connected components, and $(W_1,E(W_1)),\ldots,(W_n,E(W_n))$ be the remaining weakly connected components that are not strongly connected (where nodes that were unmatched form a trivial component by themselves).} 
\gp{Since each node touches at most two edges of the matching,} each $V_i$ corresponds to a cycle while each $W_j$ corresponds to a path, \gp{and they are all disjoint, yielding a PC cover of $G$}.

\begin{theorem}\label{thm:minpaths}
    Let $G=(V,E)$ be a node-centric de Bruijn graph, and $F\subseteq E$ be the edge set found by a maximum matching. Then, among the PC covers of $G$, the one induced by $G_F=(V,F)$ minimizes the number of open paths. 
\end{theorem}
\begin{proof}
    We classify vertices $v\in V$ into 4 types with respect to matching state of $f_L(v)\in V_L$ and $f_R(v)\in V_R$. A vertex $v$ is LR-type iff there exist edges $e'_1=(f_L(v),u),e'_2=(u',f_R(v))\in M'$ (here, $u \neq u'$ because of the matching). \gp{That is, a vertex is of LR type if both of its copies, in the left and right side, were matched.} A vertex $v$ is LX-type iff there exist an edge $e'_1=(f_L(v),u)\in M'$ and $\forall e'=(u',u'')\in M',u''\neq f_R(v)$ (\gp{Only his copy on the left was matched)}. Analogously, a vertex $v$ is XR-type iff there exist an edge $e'_2=(u,f_R(v))\in M'$ and $\forall e'=(u',u'')\in M',u'\neq f_L(v)$ (\gp{Only his copy on the right was matched)}. A vertex $v$ is XX-type iff $\forall e'=(u,u')\in M',u\neq f_L(v)$ and $u'\neq f_R(v)$ (\gp{The node was unmatched on both sides)}. LX-type and XX-type vertices have no incoming edge in $G_F$. That is, they are the first vertices of open paths. Let $\nu$ be their number: as the number of LR-type and XR-type vertices equals to the matching size $|M'|$, we have $|V|=|M'|+\nu$. Since $M'$ is maximized, the number of open paths $\nu$ is minimized.
\end{proof}

Another possible way of obtaining a minimum PC cover could be by taking Eulertigs and ``closing obvious cycles'' (i.e., if a produced string for the SPSS was actually a cycle, consider it as such instead of as a path).

\paragraph*{Separator-based Minimum PC cover}
In the separator-based model, we need to take into account the number of cycles as well. Intuitively, cycles of the PC cover will yield closed necklaces; as such, we need to minimize them as well to go towards optimization of the cost given in Lemma~\ref{lemma:BP-representation-size-separators}. In this sense, we need to find a cover of minimum total size, that is, with the minimum number of both paths and cycles. Such a cover could be obtained using the output of Eulertigs and closing any obvious cycle (i.e., check if a given path was actually a cycle in the node-centric deBruijn graph). Indeed, Eulertigs finds the minimum SPSS based on the minimum possible number of breaking edges to make the graph Eulerian, each of which creates a new cycle/path. 

\subsection{Our Algorithm\gp{: \greedyNeckAlgo}}
\label{section:algorithm}
We are now ready to present our algorithm.
\gp{As discussed before, each primitive node of the graph necessarily introduces one primitive path. In turn, each such path must introduce an open necklace (as its root), and no algorithm can avoid this. Our greedy strategy is to ensure that primitive paths are the only sources of openness, while every other path is attached to an existing cycle or primitive path. At the same time, we want to minimize the number of leaves in the pendants, as this shortens the previously-introduced parenthesis representation. }

We present \gp{\greedyNeckAlgo (as shown in Algorithm~\ref{alg:find_tree}; additional pseudocode in Appendix~\ref{sec:additional-pseudocode})}, which takes as input a node-centric de~Bruijn graph $G=(V,E)$ together with its cover by cycles and open paths (without pendants): a set of cycles $\mathcal{C}$, a set of paths $\mathcal{P}$, and their adjacency lists $N$ (and inverse adjacency lists $N^{-1}$). \gp{Without loss of generality, we can assume that none of the paths can be closed into a cycle.}
The algorithm outputs a necklace cover, represented as a subset of edges $F \subseteq E$, such that the 
connected components of the induced subgraph $G_F=(V,F)$ correspond to the necklaces. Whether a necklace is open or closed can be determined directly from its in-degree in $G_F$: a necklace is closed if every vertex has in-degree one, and open otherwise \gp{(in which case, exactly one node has in-degree zero)}.

Intuitively, \gp{\greedyNeckAlgo} grows necklaces by starting from cycles and paths and iteratively attaching pendant paths as beads, repeating this process until no path remains.  
\gp{The greedy attachments guarantee that paths are absorbed in a necklace whenever their first node has an in-neighbor that belongs to that necklace. It could happen that we have a leftover set of paths that cannot attach to current necklaces. In this case, we show that part of the leftover paths contain a cycle between themselves, and they actually form a closed necklace that can be integrated without creating extra open components.}

The goal is to find a subset $F$ that minimizes the number of symbols $N_k + (k-1)N_O + 2N_L + 1$ as stated in Lemma~\ref{lemma:BP-representation-size}. \gp{We note that at each step of our algorithm, such cost is always decreased by at least 1. Indeed, when the algorithm attaches an open path to a current necklace, it is decreasing $N_O$ by one while increasing $N_L$ by one (adding a leaf) in the current cover; therefore, it is substituting the contribution of $(k-1)$ that $N_O$ had with a contribution of 2 for the new leaf. Whenever instead a cycle is detected, we are removing $h>2$ paths from the cover, decreasing $N_O$ by $h$ while adding an open necklace with $h$ leaves, thus again decreasing the cost by $(k-1)h$ but increasing it by $2h$.}

\begin{algorithm}
\caption{\greedyNeckAlgo}\label{alg:find_tree}
\KwIn{Cycles $\mathcal{C}$, paths $\mathcal{P}$, adjacency list $N$ and its inverse $N^{-1}$}
\KwOut{$F \subseteq E$ ($E$ is the edge set of node-centric de Bruijn graph)}

$F,\mathcal{P}_p,\mathcal{P}_n \gets \emptyset$ \tcp*{$\mathcal{P}_n$ is global}
\ForEach{$P = (v_1,\ldots,v_p)\in \mathcal{P}$}{
    \If{$N^{-1}(v_1) = \emptyset$}{ 
        $\mathcal{P}_p \gets \mathcal{P}_p \cup P$ \tcp*{$P$ is a primitive path}
    }
    \Else{
        $\mathcal{P}_n \gets \mathcal{P}_n \cup P$ \tcp*{$P$ is a non-primitive path}
    }
}
\tcc{build open necklaces from primitive paths}
\ForEach{$P = (v_1,\ldots,v_p)\in \mathcal{P}_p$}{ 
    \ForEach{$v_i \in P$}{
        \If{$i \neq p$}{
            $F \gets F \cup (v_i,v_{i+1})$\;
        }
        \ForEach{$u \in N(v_i)$}{
            \If{$\exists P^\prime = (v_1^\prime,\ldots,v_{p^\prime}^\prime)\in \mathcal{P}_n\:\mathrm{s.t.}\:u=v_1^\prime$}{
                $F \gets F\cup(v_i,u)\:\cup\:$\textsc{FindSubtree}$(P^\prime,N)$\;
            }
        }
    }
}
\tcc{build closed necklaces from cycles}
\ForEach{$C=(v_1,\ldots,v_c) \in \mathcal{C}$}{
    \ForEach{$v_i \in C$}{
        $F \gets F \cup (v_i,v_{(i+1)\%(c+1)})$\;
        \ForEach{$u \in N(v_i)$}{
            \If{$\exists P^\prime = (v_1^\prime,\ldots,v_{p^\prime}^\prime)\in \mathcal{P}_n\:\mathrm{s.t.}\:u=v_1^\prime$}{
                $F \gets F\cup(v_i,u)\:\cup\:$\textsc{FindSubtree}$(P^\prime,N)$\;
            }
        }
    }
}
\tcc{build closed necklaces with unused non-primitive paths}
\While{$\mathcal{P}_n\neq\emptyset$}{
    \ForEach{$P \in \mathcal{P}_n$}{
        $C\gets$\textsc{FindNewCycle}$(P)$\;
        \If{$C\neq\emptyset$}{
            \ForEach{$v_i \in C=(v_1,\ldots,v_c)$}{
                $F \gets F \cup (v_i,v_{(i+1)\%(c+1)})$\;
                \ForEach{$u \in N(v_i)$}{
                    \If{$\exists P^\prime = (v_1^\prime,\ldots,v_{p^\prime}^\prime)\in \mathcal{P}_n\:\mathrm{s.t.}\:u=v_1^\prime$}{
                        $F \gets F\cup(v_i,u)\:\cup\:$\textsc{FindSubtree}$(P^\prime,N)$\;
                    }
                }
            }
            \textbf{break}\;
        }
    }    
}
\Return $F$\;
\end{algorithm}

The algorithm proceeds in two main phases. First, it classifies the paths of $\mathcal{P}$ into two categories: {primitive paths} (\gp{whose first node is primitive}) 
and {non-primitive paths}, which can potentially be attached elsewhere.  
Primitive paths serve as the backbones of open necklaces and are stored in a set $\mathcal{P}_p$, whereas the remaining paths are stored in $\mathcal{P}_n$.  
Next, the algorithm recursively attaches the non-primitive paths in $\mathcal{P}_n$ to existing structures; namely, either to vertices of cycles, forming closed necklaces, or to vertices of primitive paths, extending open necklaces.  
This attachment process is performed by a subroutine called \textsc{FindSubtree}, which attempts to connect each path $P'=(v'_1,\ldots,v'_{p'})$ to a vertex $u$ of the current structure by adding the edge $(u,v'_1)$ together with all edges of $P'$. (For the sake of completeness, we report additional pseudocode in Appendix~\ref{sec:additional-pseudocode}.)

\begin{figure}[htbp]
    \centering
    \includegraphics[width= 0.75\linewidth]{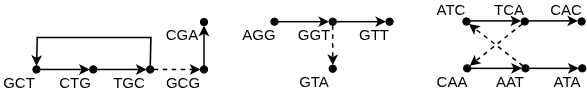}
    \caption{Example where \textsc{FindNewCycle} must be executed on the paths to the right\gp{: the two paths corresponding to \A\T\C\A\C~and \C\A\A\T\A~can be transformed into a closed necklace with base cycle \A\T\C\A\A~and pendants \C\A\C, \A\T\A}.}
    \label{fig:alg2needed}
\end{figure}

After all possible attachments to existing cycles and primitive paths have been made, it may still happen that some non-primitive paths remain unattached (see Figure~\ref{fig:alg2needed} for an example).  
In this case, the algorithm invokes the subroutine \textsc{FindNewCycle} (see Appendix~\ref{sec:additional-pseudocode}), which performs a depth-first traversal to detect a new cycle passing through vertices covered by the remaining paths.  
When such a cycle is found, it is added to $F$, and each path intersecting it is truncated to remove the portion overlapping with the cycle.  
The remaining suffixes are then recursively attached as pendants to this newly discovered cycle.  
This cycle-discovery and attachment process is repeated until all non-primitive paths have been incorporated. Lemma~\ref{lem:exist_cycle} guarantees that, as long as non-primitive paths remain, such a new cycle can always be found.

\begin{lemma}[Existence of a new cycle]\label{lem:exist_cycle}
During the execution of Algorithm~\ref{alg:find_tree}, suppose that some non-primitive paths remain, i.e. $\mathcal{P}_n \neq \emptyset$. Then there must exist a cycle $C=(v_1,\ldots,v_c)$ in $G$ with the property that every vertex $v_i$ of $C$ appears in at least one of the remaining non-primitive paths in $\mathcal{P}_n$.
\end{lemma}
\begin{proof}
    Take any $P_1\in\mathcal{P}_n$. Due to greediness of Algorithm \ref{alg:find_tree}, \gp{since such path is not primitive but cannot be attached to other necklaces,} we are sure that $\forall u\in N^{-1}(P_1\mathrm{'s \:head}),\exists P\in\mathcal{P}_n\:\mathrm{s.t.}\:u\in P$. Take any such $u\in P\in\mathcal{P}_n$ and let $P_2=P$ and $h_2$ s.t. $P_2[h_2]=u$. We repeat this until $\exists j\in [1,i]\:\mathrm{s.t.}\:P_{i+1}=P_j$, recording each $h_i$ at the same time. In the worst case, $i$ would add up to $i=|\mathcal{P}_n|$, where $P_{i+1}$ must equate to one of $P_1,\ldots,P_i$ because otherwise the greediness is violated. Then, $P_i[1],\ldots,P_i[h_i],P_{i-1}[1],\ldots,P_{i-1}[h_{i-1}],\ldots,P_j[1],\ldots,P_j[h_j]$ is a cycle.
\end{proof}

At the end of the construction, every vertex of $G$ belongs to either a cycle or to a \emph{primitive} path extended with pendant trees, and the induced subgraph $G_F$ thus represents a collection of disjoint necklaces. Hence, $F$ defines a necklace cover of the original graph. The algorithm runs in quadratic time using linear space, but by maintaining a visit status for vertices it can be implemented in $O(|V|+|E|)$ time with $O(|V|)$ additional memory.

\gp{We prove that our algorithm correctly outputs a minimum necklace cover whenever the input PC cover is minimum (i.e., it has minimum number of paths):}

\begin{theorem}
\label{thm:halt_and_optimal}
Let $\mathcal{P} \cup \mathcal{C}$ be a minimum PC cover. Then Algorithm~\ref{alg:find_tree} outputs a necklace cover $F\subseteq E$ that minimizes the number of symbols $N_k + (k-1)N_O + 2N_L + 1$ (as stated in Lemma~\ref{lemma:BP-representation-size}). Its running time is $O(|V|+|E|)$.
\end{theorem}
\begin{proof}
As previously mentioned, the number $N_O$ of open necklaces in any cover must be at least the number of primitive nodes, since each primitive path necessarily induces one.
\gp{By minimality of the cover, the number of starting non-primitive paths is minimized (primitive paths are a fixed number given by the graph's topology), and, furthermore, no two paths in the set $\mathcal{P}$ can be concatenated (the last node of a path is never an in-neighbor of the starting node of another path). Because of this and of Lemma~\ref{lem:exist_cycle},}
at every iteration of Algorithm~\ref{alg:find_tree}, both when we perform the attachment operation or when we discover cycles, the number of non-primitive paths decreases by at least one. Consequently, $N_O$ is reduced by at least one, and the value of the formula decreases by at least $k-1$. For each such decrement, one new leaf is created (in both cases of attaching the path or creating a new cycle with paths attached, the last nodes of the paths become leaves), adding a cost of $2$. Since $k \geq 2$, the total number of symbols does not increase (and surely decreases for $k \geq 3$). 
At the end, $N_O$ equals the number of primitive paths, which is optimal. 
\gp{Since the starting number of paths of $\mathcal{P}$ was minimal (the input PC cover is minimal)}, the produced value of $N_L$, which is equal to the number of said paths that are not primitive, is also optimal.
\end{proof}

Theorem~\ref{thm:halt_and_optimal} shows that our greedy algorithm always halts and, \gp{given a minimum PC cover}, produces an optimal necklace cover. Note that minimizing the number $N_L$ of leaves alone can be achieved by a simple greedy algorithm. In contrast, Algorithm~\ref{alg:find_tree} minimizes the entire formula of Lemma~\ref{lemma:BP-representation-size}.

Armed with our greedy Algorithm~\ref{alg:find_tree}, we can obtain an optimal necklace cover by applying it to the minimum PC cover obtained in Section~\ref{section:min-pc-cover}. 
The dominant time complexity in the overall time cost is that stated in Theorem~\ref{thm:minpaths}.

\paragraph*{Separator-based representation}
If we take into account the separators then the cost of a necklace cover changes, and the number of closed necklaces start to matter (see Lemma~\ref{lemma:BP-representation-size-separators}).

In this case, we can show that our algorithm is optimal if the input PC cover minimizes the separator-aware value $N_O+N_C$, instead of only minimizing the number of paths:
\begin{theorem}\label{thm:halt_and_optimal-separators}
    Suppose that the input PC cover minimizes the number of paths and cycles.
    Then Algorithm~\ref{alg:find_tree} outputs a necklace cover $F\subseteq E$ that minimizes the number of symbols $N_k + kN_O + 2N_L + N_C$ (as stated in Lemma~\ref{lemma:BP-representation-size-separators}). Its running time is $O(|V|+|E|)$.
\end{theorem}
\begin{proof}
    First, we observe that  the number of open necklaces $N_O$ in any cover must be at least the number of primitive paths, since each primitive path necessarily induces one.
    By Lemma~\ref{lem:exist_cycle}, at every iteration of Algorithm~\ref{alg:find_tree}, the number of non-primitive paths decreases by at least one. Consequently, $N_O$ is reduced by at least one, and the value of the formula decreases by at least $k$. For each such decrement, one new leaf is created, additionally one new cycle may be created, thus adding a cost of $2$ or $3$, respectively. Since $k \geq 3$, the total number of symbols does not increase (and surely decreases for $k \geq 4$). At the end, $N_O$ equals the number of primitive paths, which is optimal.  
    Second, we observe that all non-primitive paths become leaves or parts of cycles. After $N_O$ becomes optimal, $N_L$ is also optimal (given $N_O$) because the set of paths is irreducible, and each non-primitive path must become a leaf (it cannot be joined to another path) if it is not entirely subsumed by a cycle.  Finally, after $N_O$ and $N_L$ become minimal,
    $N_C$ is minimized as a new cycle is created only if a non-primitive path can only be joined with itself instead of another path. Hence, the total number of symbols is the smallest possible among all necklace covers.
\end{proof}

\subsection{Comparison with SPSS}
\label{section:eulertigs-comparison}
In this section, we show that there exist infinite families of input string sets for which our necklace cover yields smaller representations than \gp{the SPSS,} for instance given by Eulertigs.  
The key idea is to construct a set of strings whose corresponding node- and edge-centric de~Bruijn graphs consist of a cycle with single-node pendants.  
In such cases, our method produces a single closed necklace, whereas any Eulertigs-based solution requires a number of strings equal to the number of pendants. Our necklace cover requires only a fraction $4/(k+1)$ of the symbols needed by the Eulertigs representation, that is, $4n$ versus $(k+1)n$ symbols, for infinitely many values of $n = |\Sigma|^{k-2}$ with $k \geq 4$.

Given any $k$, let us consider a de~Bruijn sequence $S$ of order $k-2$ over $\Sigma$, that is, a cyclic string in which every distinct $(k-2)$-mer occurs exactly once.  
For example, for $k=4$, a de~Bruijn sequence of order $k-2=2$ is 
\texttt{ACTAGATCCGTTGGCA}.  Any such de~Bruijn sequence has length $n = |\Sigma|^{k-2}$ (and there are exponentially many distinct sequences of this length). Let $\alpha$ denote the first $k$-mer of $S$, and define the string $X = S\alpha$ as one of the input strings.  
Note that both the node-centric and edge-centric de~Bruijn graphs of order $k$ corresponding to $X$ form a simple cycle. Indeed, since all $(k-2)$-mers of $S$ are distinct, so are its $(k-1)$-mers and $k$-mers. By appending the first $k$-mer again at the end, we ensure that the new $k$-mers and $(k-1)$-mers introduced occur only at the beginning of $S$, thereby forming a cycle.  
For the above de~Bruijn sequence ($k=4$), we obtain  
$
X = \texttt{\underline{ACTA}GATCCGTTGGCA\underline{ACTA}}
$.

Next, we attach one-node pendants to each node of the cycle.  
We do this by adding $n$ additional strings to the input set as follows.  
For each $k$-mer $x_1 \cdots x_k$ of $X$, it has exactly one successor in $X$ (its out-neighbor in the node-centric de~Bruijn graph), given by some $x_2 \cdots x_k a$.  
For any $b \neq a$, we add the string $x_1 \cdots x_k b$ to the input set.  
By construction, neither these new $k$-mers nor their $(k-1)$-mers occur as substrings of any other input string.  
For instance, in the string $X$ above, the only $k$-mer (resp.\ $(k-1)$-mer) following $\texttt{TCCG}$ (resp.\ $\texttt{CCG}$) is $\texttt{CCGT}$ (resp.\ $\texttt{CGT}$).  
Hence, we can safely add $\texttt{CCGA}$ to the input set without introducing any repeated $k$- or $(k-1)$-mers.

Given an input set of strings constructed as described above, both the corresponding node- and edge-centric de~Bruijn graphs consist of a cycle of length $n$ (corresponding to string $X$), with one pendant attached to each node of the cycle (corresponding to the $k$-mers of the form $a_2\cdots a_k b$).  
Hence, both the node- and edge-centric graphs contain $2n$ nodes and $2n$ edges.  
(In general, we could have up to $(|\Sigma|-1)n$ pendants if desired.)  

Our necklace-cover representation in this case consists of a single closed necklace, using exactly $2n$ symbols and $2n$ parentheses.  
The Eulertigs representation, on the other hand, must necessarily include $n$ breaking edges (one per pendant), resulting in $(k-1)n + 2n$ symbols.  
By increasing $k$ and $n$, the advantage of our representation becomes arbitrarily large, as it is always $(k-1)n$ symbols smaller than the Eulertigs representation.  

A full example for $k=4$ is given in Figure~\ref{fig:simple-necklace-stringset}, for input string set $I= \{X=\mathtt{ACTAGATCCGTTGGCAACTA}$, $\mathtt{ACTAC}$, $\mathtt{CTAGG}$, $\mathtt{TAGAC}, \mathtt{AGATA}, \mathtt{GATCT}, \mathtt{ATCCC}, \mathtt{TCCGG}, \mathtt{CCGTA}, \mathtt{CGTTA}, \mathtt{GTTGT}, \mathtt{TTGGA}, \mathtt{TGGCG}, \mathtt{GGCAT}$, $\mathtt{GCAAA}$, $\mathtt{CAACG}$, $\mathtt{AACTT} \}$.

\begin{figure}
    \centering 
    \includegraphics[width=7cm]{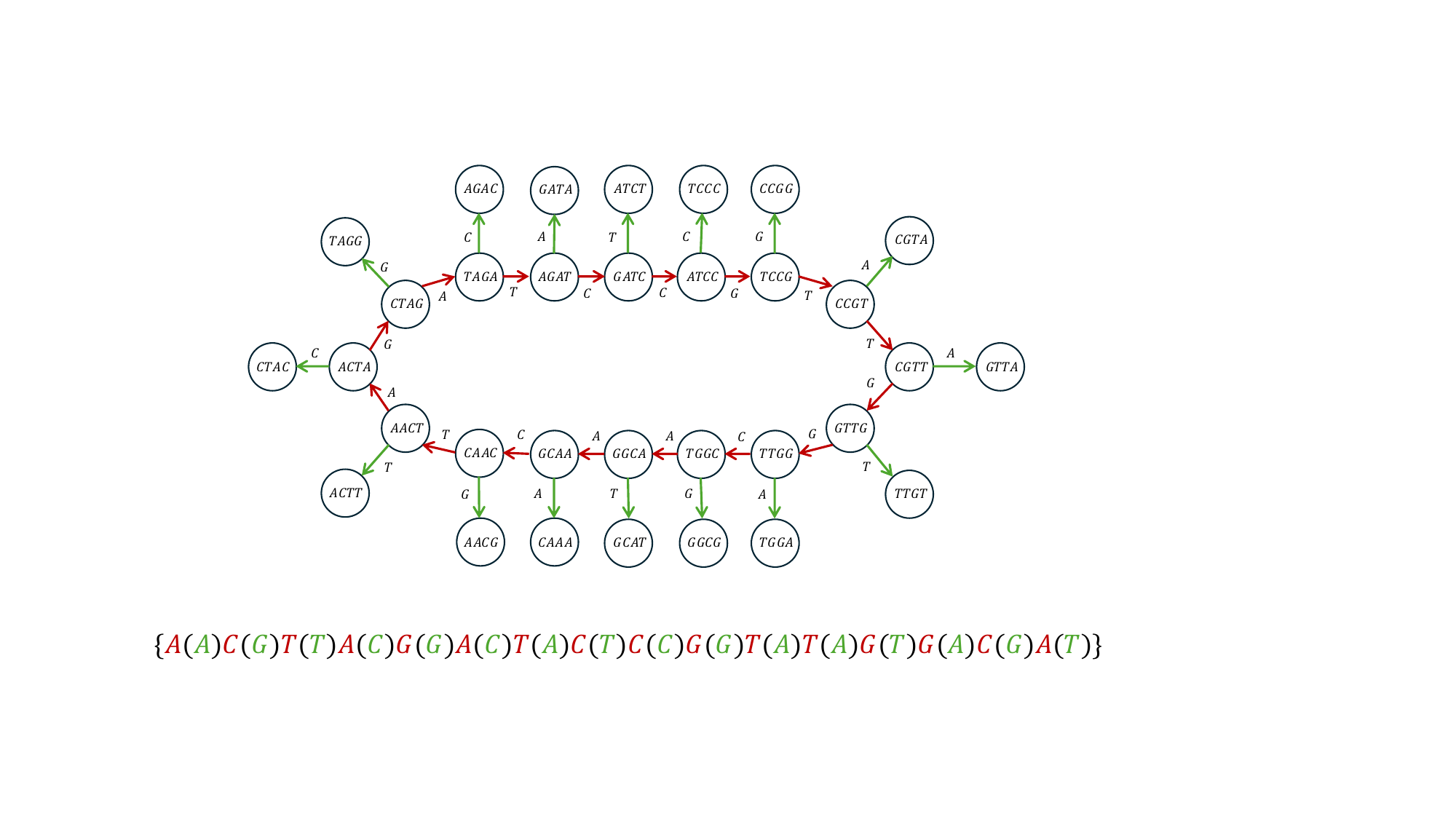} \hfill
     \includegraphics[width=6.5cm]{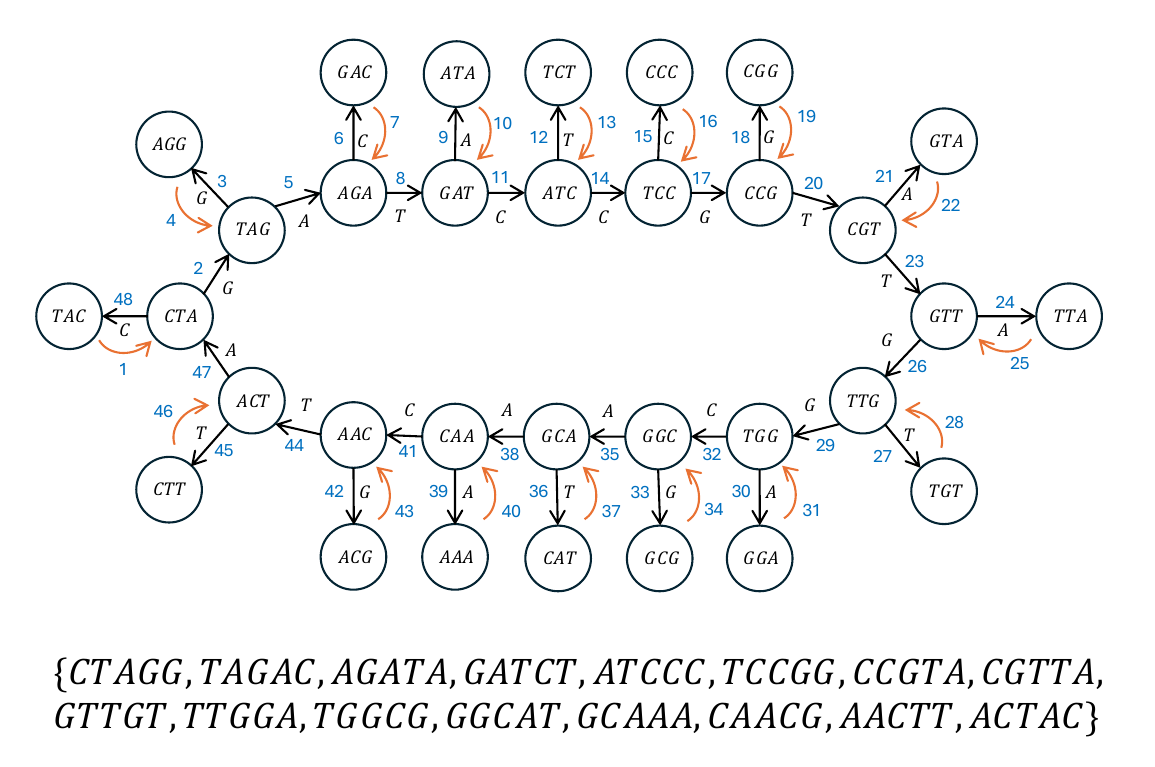}
    \caption{Consider the input string set $I = \{ X =$ \texttt{ACTAGATCCGTTGGCAACTA, ACTAC, CTAGG, TAGAC, AGATA, GATCT, ATCCC, TCCGG, CCGTA, CGTTA, GTTGT, TTGGA, TGGCG, GGCAT, GCAAA, CAACG, AACTT}$\}$. This input set of strings follows the construction of Section~\ref{section:eulertigs-comparison}, for $k=4$ and $n=|\Sigma|^{k-2}=4^2=16$. 
    For this input, we obtain a simple closed necklace as de~Bruijn graph both in the node-centric (top) and edge-centric (bottom) case. On the top, we see our simple necklace solution for this graph, with the parenthesis representation, leading to $2n= 32$ symbols, plus $2n=32$ parentheses for representation, for a total of $64$ characters. 
    On the bottom, we see an Eulertigs solution (breaking arcs in orange), leading to $n(k+1) = 80$ plaintext characters.}
    \label{fig:simple-necklace-stringset}
\end{figure}

\section{Experiments}
\label{section:experiments}
We conducted experiments to evaluate the practical impact of using necklaces in our $k$-mer representations, on a machine equipped with an Intel(R) Xeon(R) E5-1620 v4 CPU (3.50 GHz, 4 cores, 8 threads) and 185 GiB of RAM.

\smallskip
\noindent
\textbf{Datasets.}
We used four datasets in FASTA format: \textit{Chr19} (57~MB, reference), \textit{C.~Elegans} (15.3~MB, reads), \textit{B.~Mori} (493~MB, reads), and \textit{H.~Sapiens} (3.34~GB, reads).  
These datasets are available in the NCBI RefSeq database (\url{https://www.ncbi.nlm.nih.gov/datasets/genome/}) under accession IDs \texttt{GCF\_000001405.40}, \texttt{GCF\_000002985.6}, \texttt{GCF\_000151625.1}, and \texttt{GCF\_000001405.39}, respectively. 
We set $k \in \{11,13,\ldots,29,31\}$ for our study on their $k$-mers.

\smallskip
\noindent
\textbf{Compared methods.}
We compared our proposed method against \gp{a fully greedy DFS-based non-optimal variant (\greedyBaseline), and the two state-of-the-art competitors}: 
Eulertigs~\cite{schmidt2023eulertigs}
and Masked Superstring~\cite{sladky2023masked}, measuring both output size \gp{(as the cumulative length of the resulting representation, without separators)} and execution time. The implementation of \texttt{MaskedSuperstring} was obtained from \url{https://github.com/OndrejSladky/kmercamel}.\footnote{We also consider the state-of-the-art  \texttt{MaskedSuperstring}, even if it does not solve the SPSS problem exactly, as it may construct a superstring containing extra $k$-mers not present in the input set $I$.}  
We implemented \texttt{Eulertigs} and our proposed methods in C:
\implementedNeckCover is our proposed algorithm as presented in Section~\ref{section:algorithm} with input given by the minimum PC cover yielding from a maximum bipartite matching (as per Section~\ref{section:min-pc-cover}), while
\texttt{greedyBaseline} is a heuristic to directly yield necklaces, iterating depth-first search from any unvisited node in the graph until every node is visited, \gp{and later closing necklaces whenever possible.} 
The source code is available at \url{https://github.com/ren-kimura/Necklace}, along with additional implementations.

\smallskip
\noindent
\textbf{Measure of the input/output size.}
To compare the output size, we adopt as metric the number of characters over the alphabet $\Sigma \cup \{(,)\}$ used to represent the $k$-mers in the DNA sequences. We ignore the overhead introduced by the FASTA format, since it is unlikely to be identical across tools: FASTA record headers may contain different amounts of information depending on the program that generates them. When comparing space occupancy, we always refer to this measure. Accordingly, we post-process the output of the compared methods so that we count the number of produced characters under the same convention as for our proposed methods. Specifically, we count characters from the alphabet $\Sigma$ for 
\texttt{Eulertigs}, and \texttt{MaskedSuperstring}, and characters from $\Sigma \cup \{(,)\}$ for \texttt{greedyBaseline} and \implementedNeckCover. For consistency, we measure the input size simply as the total number of DNA bases in the input files.

\smallskip
\noindent
\textbf{Results.}
Table~\ref{tab:o/i}  summarizes the key statistics for compression ratios. 
As we are mainly interested in worst-case behavior, we note that the best maximum output size/input size ratios are comparably achieved by
\texttt{MaskedSuperstring} and \implementedNeckCover.

\begin{table}[htbp]
    \begin{minipage}{0.90\linewidth}
        \begin{tabular}{|l||c|c||c|c|}
            \hline
            & \greedyBaseline 
            & \implementedNeckCover 
            & \texttt{Eulertigs} 
            & \texttt{MaskedSuperstring}
            \\ \hline
            max 
            & 1.012 
            & {0.966} 
            & 1.178 
            & 0.976
            \\\hline
            min 
            & 0.482 
            & 0.234 
            & 0.348 
            & 0.232
            \\\hline
        \end{tabular} 
    \end{minipage}
    \caption{Output size/input size ratio for each method. The values correspond to the maximum and minimum ratios observed across \textit{C.Elegans} datasets (for all $k$).}
    \label{tab:o/i}
 \end{table}

Further experimental results are reported in Tables~\ref{tab:size1} and \ref{tab:size2} (output size), 
and in Tables~\ref{tab:time1} and \ref{tab:time2} (running time), for increasing values of $k$.  
Our method \implementedNeckCover achieve the best overall balance between time and space, producing smaller outputs for sufficiently large values of $k$.  
The smallest space occupancy was obtained either by \texttt{MaskedSuperstring} or by our \implementedNeckCover representation, depending on the value of $k$.  
For small $k$, \texttt{MaskedSuperstring} produced the most compact outputs, whereas \implementedNeckCover achieved comparable compression.  
For larger $k$, \implementedNeckCover achieved the smallest space occupancy among all methods.  

In terms of execution time, \texttt{MaskedSuperstring} was consistently the fastest algorithm.
It should be noted, however, that \texttt{MaskedSuperstring} does not solve the Spectrum-Preserving String Set (SPSS) problem exactly: because it constructs a superstring containing all input $k$-mers, it may also introduce \emph{false positives}, i.e., $k$-mers that do not occur in the original sequences.  
In contrast, \implementedNeckCover, as well as Eulertigs, guarantee an exact $k$-mer spectrum.  
Therefore, \implementedNeckCover offer a favorable trade-off between space efficiency and computational cost while maintaining exactness of the represented spectrum for the SPSS problem.

\begin{sidewaystable}
\small
{\setlength{\tabcolsep}{3.5pt}\setlength{\extrarowheight}{2pt}
\textit{Chr19} (58440758 symbols, chromosome)

\begin{tabular}{lrrrrrrrrrrr}
& $k=11$ & 13 & 15 & 17 & 19 & 21 & 23 & 25 & 27 & 29 & 31\\
\hline
\greedyBaseline & 9114631 & 44179571 & 47447579 & 43994652 & 44640578 & 45850545 & 47026645 & 48102585 & 49068064 & 49936216 & 50718775\\
\implementedNeckCover & \textbf{3974031} & \textbf{25983959} & \textbf{39900863} & \textbf{42266186} & \textbf{43844638} & \textbf{45242355} & \textbf{46496097} & \textbf{47647805} & \textbf{48676314} & \textbf{49597844} & \textbf{50414307}\\
\hline
\texttt{Eulertigs} & 4954903 & 46377639 & 58882919 & 51469114 & 51481358 & 52840317 & 54153977 & 55520901 & 56622402 & 57524646 & 58105179\\
\texttt{MaskedSuperstring} & 3990364 & 26092857 & 40734878 & 43672574 & 45585771 & 47187960 & 48554570 & 49767823 & 50804790 & 51668166 & 52423894\\
\hline
\end{tabular}

\medskip

\textit{C. Elegans} (15072434 symbols, read)

\begin{tabular}{lrrrrrrrrrrr}
& $k=11$ & 13 & 15 & 17 & 19 & 21 & 23 & 25 & 27 & 29 & 31\\
\hline
\greedyBaseline & 7263909 & 15254846 & 15069901 & 14511520 & 14362808 & 14366588 & 14408979 & 14457002 & 14499663 & 14536507 & 14568703\\
\implementedNeckCover & 3527615 & \textbf{10613396} & \textbf{13329033} & \textbf{14003702} & \textbf{14213522} & \textbf{14314850} & \textbf{14384813} & \textbf{14441330} & \textbf{14488093} & \textbf{14527585} & \textbf{14561497}\\
\hline
\texttt{Eulertigs} & 5241719 & 17281476 & 17757213 & 16112536 & 15286914 & 15021278 & 14961873 & 14961828 & 14967589 & 14973693 & 14978305\\
\texttt{MaskedSuperstring} & \textbf{3492706} & 10735826 & 13599853 & 14207925 & 14383726 & 14479500 & 14546470 & 14598426 & 14639866 & 14674618 & 14703609\\
\hline
\end{tabular}
\caption{Output size [symbols]}
\label{tab:size1}
}

\bigskip
\textit{Chr19} (58440758 symbols, chromosome)

\begin{tabular}{lrrrrrrrrrrr}
& $k=11$ & 13 & 15 & 17 & 19 & 21 & 23 & 25 & 27 & 29 & 31\\
\hline
\greedyBaseline & 14.61 & 44.01 & 61.29 & 66.05 & 70.03 & 70.12 & 74.8 & 74.91 & 78.3 & 76.12 & 80.36\\
\implementedNeckCover & 20.2 & 84.45 & 116.21 & 116.45 & 123.26 & 123.41 & 133.3 & 132.82 & 146.1 & 137.01 & 151.59\\
\hline
\texttt{Eulertigs} & 14.87 & 50.17 & 76.87 & 84.12 & 85.59 & 97.34 & 92.24 & 93.37 & 97.34 & 107.38 & 97.26\\
\texttt{MaskedSuperstring} & \textbf{3.95} & \textbf{22.46} & \textbf{26.62} & \textbf{22.72} & \textbf{23.04} & \textbf{24.09} & \textbf{25.2} & \textbf{26.38} & \textbf{27.21} & \textbf{27.58} & \textbf{25.8}\\
\hline
\end{tabular}

\medskip

\textit{C. Elegans} (15072434 symbols, read)

\begin{tabular}{lrrrrrrrrrrr}
& $k=11$ & 13 & 15 & 17 & 19 & 21 & 23 & 25 & 27 & 29 & 31\\
\hline
\greedyBaseline & 6.69 & 14.44 & 18.43 & 19.5 & 18.75 & 18.91 & 19.83 & 21.18 & 20.86 & 21.06 & 20.79\\
\implementedNeckCover & 11.72 & 29.36 & 34.04 & 36.46 & 32.1 & 33.77 & 32.28 & 32.14 & 35.39 & 35.48 & 37.45\\
\hline
\texttt{Eulertigs} & 7.3 & 19.15 & 23.33 & 25.14 & 26.25 & 26.29 & 26.84 & 27.36 & 25.34 & 25.69 & 25.87\\
\texttt{MaskedSuperstring} & \textbf{2.3} & \textbf{6.87} & \textbf{6.72} & \textbf{5.12} & \textbf{4.66} & \textbf{4.69} & \textbf{4.55} & \textbf{4.57} & \textbf{4.59} & \textbf{4.47} & \textbf{4.31}\\
\hline
\end{tabular}
\caption{Running time [seconds]}
\label{tab:time1}
\end{sidewaystable}

\begin{sidewaystable}
\small
{\setlength{\tabcolsep}{1.5pt}\setlength{\extrarowheight}{2pt}
\textit{B. Mori} (431723578 symbols, read)

\begin{tabular}{lrrrrrrrrrrr}
& $k=11$ & 13 & 15 & 17 & 19 & 21 & 23 & 25 & 27 & 29 & 31\\
\hline
\greedyBaseline & 10485741 & 136816027 & 335483952 & 335366591 & 321193521 & 321944384 & 326697741 & 331837062 & 336735380 & 341305892 & 345573842\\
\implementedNeckCover & \textbf{4194313} & 63001757 & \textbf{225933314} & \textbf{293233481} & \textbf{309783011} & \textbf{317694506} & \textbf{324042359} & \textbf{329703454} & \textbf{334884670} & \textbf{339664584} & \textbf{344094966}\\
\hline
\texttt{Eulertigs} & 4194385 & 94391627 & 404049470 & 414476071 & 370966275 & 361147982 & 364353879 & 370086060 & 375970294 & 381474300 & 386463782\\
\texttt{MaskedSuperstring} & 4202555 & \textbf{62108078} & 227323845 & 298953443 & 317846884 & 327155750 & 334524320 & 340975936 & 346769657 & 352025740 & 356819381\\
\hline
\end{tabular}

\medskip

\textit{H. Sapiens} (3136819154 symbols, read)

\begin{tabular}{lrrrrrrrrrrr}
& $k=11$ & 13 & 15 & 17 & 19 & 21 & 23 & 25 & 27 & 29 & 31\\
\hline
\greedyBaseline & 10482963 & 153886427 & 1239407140 & \xx & \xx & \xx & \xx & \xx & \xx & \xx & \xx\\
\implementedNeckCover & \textbf{4194861} & \textbf{64831379} & \textbf{620166986} & \textbf{1951835859} & \textbf{2308609839} & \xx & \xx & \xx & \xx & \xx & \xx\\
\hline
\texttt{Eulertigs} & 4200829 & 77130469 & 1049414570 & 3747907014 & \xx & \xx & \xx & \xx & \xx & \xx & \xx\\
\texttt{MaskedSuperstring} & 4225846 & 65188060 & {641731970} & {1966819246} & {2349162755} & \textbf{2439174597} & \textbf{2497770866} & \textbf{2545958853} & \textbf{2587140480} & \textbf{2622694500} & \textbf{2653704607}\\
\hline
\end{tabular}
\caption{Output size [symbols]\quad(Gray cells: Infeasible due to memory constraints)}
\label{tab:size2}
}

\bigskip
\textit{B. Mori} (431723578 symbols, read)

\begin{tabular}{lrrrrrrrrrrr}
& $k=11$ & 13 & 15 & 17 & 19 & 21 & 23 & 25 & 27 & 29 & 31\\
\hline
\greedyBaseline & 86.86 & 209.52 & 448.88 & 546.59 & 582.6 & 595.53 & 585.42 & 626.68 & 630.33 & 609.18 & 614.26\\
\implementedNeckCover & 89.17 & 322.59 & 893.37 & 1071.76 & 1120.43 & 1025.46 & 1103.85 & 1177.35 & 1071.12 & 1121.76 & 1150.15\\
\hline
\texttt{Eulertigs} & 82.5 & 205.63 & 535.64 & 684.12 & 728.13 & 820.13 & 817.04 & 769.04 & 848.85 & 831.18 & 860.09\\
\texttt{MaskedSuperstring} & \textbf{22.65} & \textbf{64.41} & \textbf{238.94} & \textbf{220.83} & \textbf{197.62} & \textbf{196.48} & \textbf{198.6} & \textbf{205.27} & \textbf{209.35} & \textbf{216.1} & \textbf{212}\\
\hline
\end{tabular}

\medskip

\textit{H. Sapiens} (3136819154 symbols, read)

\begin{tabular}{lrrrrrrrrrrr}
& $k=11$ & 13 & 15 & 17 & 19 & 21 & 23 & 25 & 27 & 29 & 31\\
\hline
\greedyBaseline & 488.57 & 757.39 & 1835.03 & \xx & \xx & \xx & \xx & \xx & \xx & \xx & \xx\\
\implementedNeckCover & 515.68 & 883.43 & 3268.26 & \xx & \xx & \xx & \xx & \xx & \xx & \xx & \xx\\
\hline
\texttt{Eulertigs} & 484.05 & 754.58 & 1959.49 & \xx & \xx & \xx & \xx & \xx & \xx & \xx & \xx\\
& (\textbf{153.35}) & (272.46) & (1809.6) & (6268) & \xx & \xx & \xx & \xx & \xx & \xx & \xx\\
\texttt{MaskedSuperstring} & 161 & \textbf{236.25} & \textbf{889.51} & \textbf{2331.81} & \textbf{1815.8} & \textbf{1645.45} & \textbf{1607.43} & \textbf{1650.86} & \textbf{1724.66} & \textbf{1701.09} & \textbf{1762.49}\\
\hline
\end{tabular}
\caption{Running time [seconds]\quad(Gray cells: Infeasible due to memory constraints, Parentheses: Running time using bitvector to store dBGs)}
\label{tab:time2}
\end{sidewaystable}

\textbf{Discussion.}
The main insight from these experiments is that using necklaces produced by our greedy Algorithm~\ref{alg:find_tree}, starting with a minimum PC cover 
(yielding \implementedNeckCover), achieves a space reduction comparable to or better than that of the superstring-based method \texttt{MaskedSuperstring}, but without introducing false positives in the SPSS representation. This confirms that the use of necklace structures, combining cycles and open paths with pendants, could provide an effective and exact way to achieve compact representations of $k$-mer sets.

\bibliographystyle{plain}
\bibliography{arxiv-references}

\newpage
\appendix
\section*{APPENDIX}

\section{Additional Pseudocode}
\label{sec:additional-pseudocode}

\begin{algorithm}
\caption{\textsc{FindSubtree}}\label{alg:find_subtree}
\KwIn{path $P$, adjacency list $N$, depth of recursion $d = 0$}
\KwOut{$F\subseteq E$ ($E$ is the edge set of Node-centric de Bruijn graph)}
\If{$d = 0$}{
    $F \gets \emptyset$\;
}
$\mathcal{P}_n \gets \mathcal{P}_n \setminus P$\;
\ForEach{$v_i \in P = (v_1,\ldots,v_p)$}{
    \If{$i \neq p$}{
        $F \gets F \cup (v_i,v_{i+1})$\;
    }
    \ForEach{$u \in N(v_i)$}{
        \If{$\exists P^\prime = (v_1^\prime,\ldots,v_{p^\prime}^\prime)\in \mathcal{P}_n\:\mathrm{s.t.}\:u=v_1^\prime$}{
            $F \gets F\cup(v_i,u)\:\cup\:$\textsc{FindSubtree}$(P^\prime,N)$\;
        }
    }
}
\Return $F$\;
\end{algorithm}

\begin{algorithm}
\caption{\textsc{FindNewCycle}}\label{alg:find_new_cycle}
\KwIn{start path $P_s$, current path $P$, paths $\mathcal{P}$, adjacency list $N$, depth of recursion $d = 0$}
\KwOut{cycle $C$}
\If{$d=0$}{
    $C, \mathcal{P}_C\gets\emptyset$\;
}
\If{$P=P_s \:\mathbf{and}\:d>1$}{
    \ForEach{$P^\prime\in \mathcal{P}_C$}{
        $\mathcal{P}_n \gets \mathcal{P}_n \setminus P^\prime$\;
        $P^\prime \gets P^\prime \setminus C$\ \tcp*{directly modify an element in $\mathcal{P}$}
        $\mathcal{P}_n \gets \mathcal{P}_n \cup P^\prime$\;
    }
    \Return $C$\;
}
$\mathcal{P}_C \gets \mathcal{P}_C \cup P$\;
\ForEach{$v_i\in P=(v_1,...,v_p)$}{
    $C \gets \mathrm{append}(C,v_i)$\;
    \ForEach{$u \in N(v_i)$}{
        \If{$\exists P^\prime=(v_1^\prime,\ldots,v_{p^\prime}^\prime)\in\mathcal{P}_n\:\mathrm{s.t.}\:u=v_1^\prime$}{
            \If{$P^\prime\in\mathcal{P}_C\setminus P_s$}{\textbf{continue}\;}
            $C \gets $\textsc{FindNewCycle}$(P_s,P^\prime,\mathcal{P},N,d+1)$\;
            \If{$C\neq\emptyset$}{\Return $C$\;}
        }
    }
}
\Return $\emptyset$\;
\end{algorithm}

\end{document}